%% file: main2.tex
\newif\iflongversion
\longversiontrue%

\iflongversion

  \documentclass[11pt]{article}
  \title{All-Subsets Important Separators with Applications to Sample Sets, Balanced Separators and Vertex Sparsifiers in Directed Graphs}
  \author{ {Aditya Anand \thanks{University of Michigan, Ann Arbor} }\and {Euiwoong Lee \thanks{University of Michigan, Ann Arbor. Supported in part by NSF grant CCF-2236669 and Google}}
  \and {Jason Li \thanks{Carnegie Mellon University}}
  \and {Thatchaphol Saranurak\thanks{
        University of Michigan,
        \texttt{thsa@umich.edu}.
        Supported by NSF Grant CCF-2238138. Partially funded by the Ministry of Education and Science of Bulgaria's support for INSAIT, Sofia University ``St.~Kliment Ohridski'' as part of the Bulgarian National Roadmap for Research Infrastructure.
    }}}
    \date{}
  
  \usepackage{amsthm}
  \usepackage{thm-restate}
  \newtheorem{theorem}{Theorem}[section]
\newtheorem{lemma}[theorem]{Lemma}

\newtheorem{definition}[theorem]{Definition}
\newtheorem{corollary}[theorem]{Corollary}

\usepackage[margin=1in]{geometry}
\usepackage[linktocpage=true,
pagebackref=true,colorlinks,
linkcolor=DarkRed,citecolor=ForestGreen,
bookmarks,bookmarksopen,bookmarksnumbered]
{hyperref}

\else
  \documentclass[a4paper,UKenglish,cleveref,autoref,thm-restate]{lipics-v2021}
  \title{All-Subsets Important Separators with Applications to Sample Sets, Balanced Separators and Vertex Sparsifiers in Directed Graphs}
 \titlerunning{All-Subsets Important Separators with Applications}
 \author{Aditya Anand}{University of Michigan Ann Arbor, USA}{adanand@umich.edu}{}{}
\author{Euiwoong Lee}{University of Michigan Ann Arbor, USA}{euiwoong@umich.edu}{}{}
\author{Jason Li}{Carnegie Mellon University, USA}{jmli@alumni.cmu.edu}{}{}
\author{Thatchaphol Saranurak}{University of Michigan Ann Arbor, USA}{thsa@umich.edu}{}{}
\authorrunning{A. Anand, E. Lee, J. Li and T. Saranurak}

\Copyright{Aditya Anand, Euiwoong Lee, Jason Li and Thatchaphol Saranurak}

\ccsdesc[500]{Theory of computation~Fixed parameter tractability}
\ccsdesc[500]{Theory of computation~Graph algorithms analysis}
\keywords{directed graphs, important separators, sample sets, balanced separators, vertex sparsifiers}
\category{}                                          %
\relatedversion{}                                    %
\funding{Euiwoong Lee is supported in part by NSF grant CCF-2236669 and Google.}    
\EventEditors{Keren Censor-Hillel, Fabrizio Grandoni, Joel Ouaknine, and Gabriele Puppis}
\EventNoEds{4}
\EventLongTitle{52nd International Colloquium on Automata, Languages, and Programming (ICALP 2025)}
\EventShortTitle{ICALP 2025}
\EventAcronym{ICALP}
\EventYear{2025}
\EventDate{July 8--11, 2025}
\EventLocation{Aarhus, Denmark}
\EventLogo{}
\SeriesVolume{334}
\ArticleNo{35}
\fi

\newif\iflongversion
\longversiontrue %

\newcommand{\lv}[1]{\iflongversion #1 \fi}  %

\usepackage{tikz}
\usetikzlibrary{calc,fit,positioning,arrows.meta,backgrounds}

\usepackage{graphicx} %
\usepackage{xcolor}
\definecolor{ForestGreen}{rgb}{0.1333,0.5451,0.1333}
\definecolor{DarkRed}{rgb}{0.65,0,0}
\definecolor{Red}{rgb}{1,0,0}
\lv{}

\usepackage{enumerate}

\usepackage{tikz}
\usetikzlibrary{shapes, arrows.meta, positioning}
\usepackage{amssymb}
\usepackage{amsmath}
\usepackage{thm-restate}

\usepackage{algpseudocode}
\usepackage{algorithm}
\usepackage{float}
\usepackage{tablefootnote}
\usepackage{longtable}

\usepackage{xspace}
\usepackage[colorinlistoftodos,textsize=tiny,textwidth=2cm,color=red!25!white,obeyFinal]{todonotes}
\usepackage{cleveref}

\ifdefined\DEBUG
    \newcommand{\thnote}[1]{\todo[color=red!25!white]{TS: #1}\xspace}
 
 \newcommand{\enote}[1]{\todo[color=green!25!white]{EL: #1}\xspace}
  \newcommand{\anote}[1]{\todo[color=blue!25!white]{AA: #1}\xspace}
\else
\newcommand{\enote}[1]{}
\newcommand{\anote}[1]{}
\newcommand{\thnote}[1]{}
\fi

\newcommand{\Reach}{\mathrm{Reach}\xspace}

\renewcommand{\SS}{\mathcal{S}}
\newcommand{\OO}{\mathcal{O}}

\newcommand{\DBS}{{\textsc{Directed Bisection}}{}}
\newcommand{\MB}{{\textsc{Minimum Bisection}}{}}
\newcommand{\BS}{{\textsc{Balanced Separator}}}
\newcommand{\OC}{{\textsc{Oneway Cuts}}}

\newcommand{\DB}{{\textsc{Directed Balanced Separator}}}
\newcommand{\SMC}{{\textsc{Skew Separator}}}
\newcommand{\DFAS}{{\textsc{Directed Feedback Arc Set}}}
\newcommand{\SMCE}{{\textsc{Skew Edge Separator}}}

\newcommand{\Wahlstrom}{Wahlstr\"{o}m\xspace}

\newcommand{\defparproblem}[4]{
  \vspace{1mm}
\begin{center}
\noindent\fbox{

  \begin{minipage}{0.85\textwidth}
  \begin{tabular*}{\textwidth}{@{\extracolsep{\fill}}lr} \textsc{#1}  & {\bf{Parameter:}} #3 \\ \end{tabular*}
  {\bf{Input:}} #2  \\
  {\bf{Question:}} #4
  \end{minipage}
 
  }
  \end{center}
 }

\begin{document}

\maketitle
\pagenumbering{gobble}

\input{abstract}
\lv{
\newpage
\tableofcontents
\newpage
}

\pagenumbering{arabic}

\input{intro}

\input{result}

\input{technical}

\bibliographystyle{alpha}
\bibliography{references}

\input{appendix}

\end{document}

%% file: abstract.tex
\begin{abstract}
Given a directed graph $G$ with $n$ vertices and $m$ edges, a parameter $k$ and two disjoint subsets $S,T \subseteq V(G)$, we show that the number of \emph{all-subsets important separators}, which is the number of $A$-$B$ important vertex separators of size at most $k$ over all $A \subseteq S$ and $B \subseteq T$, is at most $\beta(|S|, |T|, k) = 4^k {|S| \choose \leq k} {|T| \choose \leq 2k}$, where ${x \choose \leq c} = \sum_{i = 1}^c {x \choose i}$, and that they can be enumerated in time $\OO(\beta(|S|,|T|,k)k^2(m+n))$. This is a generalization of the folklore result stating that the number of $A$-$B$ important separators for two fixed sets $A$ and $B$ is at most $4^k$ (first implicitly shown by Chen, Liu and Lu Algorithmica '09). From this result, we obtain the following applications: 

\begin{enumerate}
    \item We give a construction for detection sets and sample sets in directed graphs, generalizing the results of Kleinberg (Internet Mathematics' 03) and Feige and Mahdian (STOC' 06) to directed graphs. 
    \item Via our new sample sets, we give the first FPT algorithm for finding balanced separators in directed graphs parameterized by $k$, the size of the separator. Our algorithm runs in time $2^{\OO(k)} \cdot (m + n)$. 
    \item Additionally, we show a $\OO(\sqrt{\log k})$ approximation algorithm for finding balanced separators in directed graphs in polynomial time. This improves the best known approximation guarantee of $\OO(\sqrt{\log n})$ and matches the known guarantee in undirected graphs by Feige, Hajiaghayi and Lee~(SICOMP' 08).
    \item Finally, using our algorithm for listing all-subsets important separators, we give a deterministic construction of vertex cut sparsifiers in directed graphs when we are interested in preserving min-cuts of size upto $c$ between bipartitions of the terminal set. Our algorithm constructs a sparsifier of size $\OO\left({t \choose \leq 3c}2^{\OO(c)}\right)$ and runs in time $\OO\left({t \choose \leq 3c} 2^{\OO(c)}(m + n)\right)$, where $t$ is the number of terminals, and the sparsifier additionally preserves the set of important separators of size at most $c$ between bipartitions of the terminals.
\end{enumerate}

\end{abstract}

%% file: intro.tex
\newpage
\section{Introduction}

The study of parameterized algorithms, or fixed-parameter tractability (FPT) (see~\cite{cygan2015parameterized,downey2013fundamentals}) along with the classical area of approximation algorithms, has emerged as one of the most promising ways to cope with NP-completeness. Given a decision problem $Q$ with input size $n$, together with a parameter $\ell$, a parameterized (or FPT) algorithm is an algorithm running in time $f(\ell) n^{\OO(1)}$ that decides $Q$. Over the last few years, parameterized algorithms have been studied for most well-studied NP-complete problems. A particular focus has been on graph cut optimization problems, including {\textsc{Multiway Cut}}, {\textsc{Multicut}}, {\textsc{Minimum Bisection}}, {\textsc{Feedback Vertex Set}}, {\textsc{Balanced Separator}}~\cite{marx2006parameterized,fm06,chen2008fixed,marx2011fixed,cygan2014minimum,cygan2020randomized,li2022detecting}.
One of the most important tools for graph cut problems has been the technique of \emph{important separators} (see~\cite{marx2011important} for a brief survey), which has formed a building block in parameterized algorithms for many of these and other problems~\cite{marx2006parameterized,chen2008fixed,chen2009improved, chitnis2013fixed, lokshtanov2013clustering,lokshtanov2021fpt}.

In this work, our contribution is two-fold. First, we show a new structural result on important separators. Next, we show that using this result, combined with other techniques, leads to interesting consequences. We show that one can compute \emph{sample sets} in directed graphs. This in turn allows us to obtain both an FPT algorithm and an improved approximation algorithm for finding small balanced separators in directed graphs. We also show a deterministic construction of vertex cut sparsifiers in directed graphs. All our algorithms are simple and concise modulo standard results in approximation and parameterized algorithms.

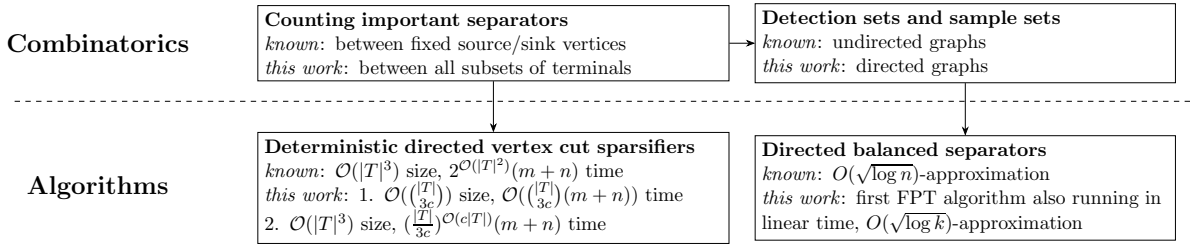
\begin{figure}[H]

\input{figure_contribution}

\caption{Summary of contribution of our paper}

\end{figure}

\subsection{All-Subsets Important Separators}
Important separators have proved to be a very important tool in the design of many FPT algorithms, including fundamental problems such as {\textsc{Multiway Cut}}, {\textsc{Multicut}}, {\textsc{Directed Feedback Arc Set}}.
We refer to Chapter 8 of~\cite{cygan2015parameterized} for various applications of important separators to design parameterized algorithms. Given two disjoint subsets of vertices $A, B$ and another set of vertices $X$ (which may intersect $A,B$) in a directed graph, we say that $X$ is an $A$-$B$ separator if there is no directed path from any vertex of $A$ to any vertex of $B$ in $G \setminus X$\footnote{$G \setminus X$ denotes the graph obtained from $G$ by deleting the vertex set $X$ along with its incident edges}. For such a separator $X$, define $R(X)$ to be the set of vertices reachable from $A$ in $G \setminus X$. Then an $A$-$B$ separator $X$ is called an important $A$-$B$ separator if it is (a) inclusion-wise minimal - so that for any $v \in X$, $X \setminus \{v\}$ is not an $A$-$B$ separator and (b) for any other $A$-$B$ separator $Y \subseteq V(G)$ with $|Y| \leq |X|$, we do not have $R(X) \subset R(Y)$.\footnote{Throughout the paper, we use the notation $P \subset Q$ to mean that $P$ is a proper subset of $Q$.} Informally, important separators are those for which there is no other separator which is further away from $A$ without increasing the cut-size. Marx~\cite{marx2006parameterized} first introduced the notion of important separators and showed a bound of $4^{k^2}$ on the number of important $A$-$B$ separators of size at most $k$. Later this bound was improved to $4^k$ (implicit in the FPT algorithm for \textsc{Vertex Multiway Cut} by Chen, Liu and Lu~\cite{chen2009improved})\footnote{While these results focused mostly on the undirected version, the same proof works in directed graphs.}. Over the last few years, this result has been used extensively to obtain FPT algorithms for various central parameterized problems~\cite{marx2011fixed, chitnis2013fixed,cygan2014minimum,lokshtanov2021fpt}.
In this work, we show a new structural result on important separators in directed graphs: Given a directed graph $G$ and two disjoint subsets $S,T \subseteq V(G)$, we show that the number of $A$-$B$ important separators across all $A \subseteq S$ and $B \subseteq T$ is at most $\beta(|S|, |T|, k) = 4^k {|S| \choose \leq k} {|T| \choose \leq 2k}$, where ${n \choose \leq a} = \sum_{i = 1}^{a} {n \choose i}$. Note that the trivial bound is $4^k 2^{|S| + |T|}$, which follows directly from the fact that there are at most $4^k$ important $A$-$B$ separators for any fixed $A \subseteq S, B \subseteq T$. Previously, no such result is known even for undirected graphs. 

Throughout this paper, given a graph, we will denote by $n$ the number of vertices and by $m$ the number of edges/arcs. 

\begin{restatable}[All Subsets Important Separators]{theorem}{impsep}\label{thm:importantbound}
Let $G$ be a digraph, $k$ be a positive integer and let $S \subseteq V(G)$ and $T \subseteq V(G)$ be disjoint sets of source and sink vertices. Then there are at most $ \beta(|S|, |T|, k) = 4^k {|S| \choose \leq k}{|T| \choose \leq 2k}$ $A$-$B$ important separators of size $\leq k$ across all $A \subseteq S$ and $B \subseteq T$. Further, they can be enumerated in time $\OO(\beta(|S|, |T|, k)\cdot k^2\cdot(m + n))$. 
\end{restatable}

It is not difficult to show that this result is essentially tight: we show this formally in the next lemma, whose proof is deferred to the appendix. 

\begin{restatable}{lemma}{lowerbound}
For any positive integer $k$, the following two statements hold.
\begin{enumerate}
\item There exist infinitely many positive integers $c$, such that for each $c$, there is a directed graph $G_{c,k}$ and disjoint subsets $S, T \subseteq V(G)$ of vertices with $|T| = c$ and $|S| = 1$ so that there are at least ${|T| \choose \leq k}$ $A$-$B$ important separators of size at most $k$ across all choices $A \subseteq S$, $B \subseteq T$.

\item There exist infinitely many positive integers $c$, such that for each $c$, there is a directed graph $G_{c,k}$ and disjoint subsets $S,T \subseteq V(G)$ of vertices with $|S| = c$ and $|T| = k + 1$ so that there are at least ${|S| \choose \leq k}$ $A$-$B$ important separators of size at most $k$ across all choices $A \subseteq S$, $B \subseteq T$.

\end{enumerate}
\end{restatable}

In  the next few subsections, we show that this simple structural result has various interesting consequences, allowing us to obtain many new results which were previously only known for undirected graphs. 

\subsection{Detection Sets and Sample Sets in Directed Graphs} 

Kleinberg~\cite{kleinberg2004detecting} introduced the concept of a detection set in a graph. The principal motivation for this concept was that of a network failure:  given a network, can we compute a small set of representative nodes so that for any  small set of failure nodes that cuts communication between two large subsets of nodes, there are two representatives that cannot communicate? Formally, Kleinberg defined an $(\epsilon, k)$ detection set for an undirected graph $G$ as a set of terminals $T \subseteq V(G)$ which satisfies the following property. First, define a network failure as a set of 
vertices $X$ with $|X| \leq k$ so that $G \setminus X$ can be partitioned into $A \cup B$ where $|A|, |B| \geq \epsilon n$ and there are no edges between $A$ and $B$. Then for every such (vertex) failure set $X$, the set $T$ must intersect with at least two components of $G \setminus X$. Kleinberg~\cite{kleinberg2004detecting} showed that there is a detection set of size $\OO(\frac{k^3}{\epsilon} \log \frac{1}{\epsilon})$. For the edge failure version, he showed a bound of $\OO(\frac{k}{\epsilon} \log \frac{1}{\epsilon})$ which was subsequently improved to $\OO(\frac{k}{\lambda}\frac{1}{\epsilon} \log \frac{1}{\epsilon})$ by Kleinberg et al.~\cite{kleinberg2008network}, where $\lambda$ is the size of the global minimum cut in the graph. Feige and Mahdian~\cite{fm06} showed an improved bound on (vertex) detection sets: they show a bound of $\OO(\frac{k}{\epsilon} \log \frac{1}{\epsilon})$. Further, they showed a bound of $\OO(\frac{k}{\epsilon})$ for the edge failure version, removing the dependence on $\log \frac{1}{\epsilon}$.

Feige and Mahdian~\cite{fm06} also studied a strengthening of this notion of detection sets called \emph{sample sets} for undirected graphs: on a high level, these are a small set of terminals $T$ which represent all the small cuts of the graph in proportion to their size, up to some additive error. Concretely, given an undirected graph $G$ and a parameter $k$, they show that there exists a set of terminals $T$ with $|T| = \OO(\frac{k}{\epsilon^2} \log \frac{1}{\epsilon})$, such that for any vertex set $X$ of size $|X| \leq k$ and every connected component $C$ of $G \setminus X$, we have $||C| - \frac{n}{|T|} |C \cap T|| \leq \epsilon n$. Further, the set $T$ can be obtained by simple random sampling, and~\cite{fm06} shows that any random subset $T \subseteq V(G)$ of size $\OO(\frac{k}{\epsilon^2} \log \frac{1}{\epsilon})$ is a sample set with constant probability.

The most important feature of these results is that the size of the detection set (or sample set) $T$ does not depend on $n$. It is then natural to ask if such results are possible for directed graphs. Given a digraph $G$, let us define a network failure as a set of vertices $X$ with $|X| \leq k$, so that $G \setminus X$ can be partitioned into two parts $A$ and $B$, each of size at least $\epsilon n$, such that there is no arc from $B$ to $A$. The analogous question for detection sets is then: is there a set of terminals $T$ with $|T| = f(k, \epsilon)$, such that for any network failure $X$ there exists a pair $t_1, t_2 \in T$ so that there is no path from $t_1$ to $t_2$ in $G \setminus X$? Similarly, the analogous question for sample set becomes: is there a set of terminals $T$ with $|T| = g(k, \epsilon)$ for some function $f$, so that for any set of vertices $X$ with $|X| \leq k$, and any \emph{strongly connected component (SCC)} $C$ of $G \setminus X$, we have $||C| - \frac{n}{|T|} |C \cap T|| \leq \epsilon n$? We answer both these questions in the affirmative, showing that one can in fact asymptotically match the same bound as in the undirected case, with $f(k, \epsilon) = \OO(\frac{k}{\epsilon} \log \frac{1}{\epsilon})$ and $g(k, \epsilon) = \OO(\frac{k}{\epsilon^2} \log \frac{1}{\epsilon})$.

\begin{restatable}{theorem}{detectionset}\label{thm:detection_set}
For any directed graph $G$, given parameters $\epsilon \in (0,1)$ and $k \in \mathbb{N}$, there is an absolute constant $c$ such that there is an $(\epsilon,k)$ detection set of size $f(k,\epsilon) = c\frac{k}{\epsilon} \log \frac{1}{\epsilon}$. Further, a random set of $f(k,\epsilon)$ vertices must be an $(\epsilon,k)$ detection set with probability at least $\frac{2}{3}$.
\end{restatable}

\begin{restatable}{theorem}{sampleset}\label{thm:sample_set}
For any directed graph $G$, given parameters $\epsilon \in (0,1)$ and $k \in \mathbb{N}$, there is an absolute constant $c$ such that there is an $(\epsilon,k)$ sample set of size $f(k,\epsilon) = c\frac{k}{\epsilon^2} \log \frac{1}{\epsilon}$. Further, a random set of $f(k,\epsilon)$ vertices must be an $(\epsilon,k)$ sample with probability at least $\frac{2}{3}$.
\end{restatable}

While we believe that this is of independent interest, we also show a similar connection to parameterized and approximation algorithms for finding balanced cuts in directed graphs along the lines of the undirected case as in~\cite{fm06}, as discussed in the following subsections. In fact, by using a slightly more nuanced analysis, our results generalize that of~\cite{fm06} even for undirected graphs. 

Finally, we note that one can easily prove that there exists an absolute constant $c$ such that a random subset of size $c\frac{k \log n}{\epsilon^2}$ is a sample set with constant probability in a directed graph. This follows from a simple application of Chernoff bounds and a union bound noting that the number of vertex sets of size at most $k$ is at most ${n \choose k}$ (${m \choose k}$ in the case of edge sets). However, for most applications this bound is too weak. For instance, our parameterized algorithm for \DB{} has an exponential dependency in the size of the sample set, and hence to show results parameterized by $k$, it is essential that the size of the sample set does not depend on $n$.

\subsection{Directed Balanced Cuts} One of the most well-studied graph partitioning problems, both from the parameterized and approximation algorithms point of view, is the problem of \MB{}. Given an undirected graph $G$ and a parameter $k$, the goal of the \MB{} problem is to obtain a partition of the graph into two equal sized parts, such that the number of cut edges is at most $k$. \MB{} was shown to be fixed-parameter tractable by Cygan et al.~\cite{cygan2014minimum} and the current best parameterized algorithm is due to~\cite{cygan2020randomized} who show an algorithm with running time $2^{\OO(k \log k)}n^{\OO(1)}$. Räcke~\cite{racke2008optimal} gave an $\OO(\log n)$ approximation for (the optimization version of) \MB. If only an approximate bisection where both sides have $\Omega(n)$ vertices is desired, then this problem is essentially the \BS{} problem which is known to have an FPT algorithm running in time $2^{\OO(k)}(m + n)$ due to Feige and Mahdian~\cite{fm06} and an $\OO(\sqrt{\log n})$ approximation algorithm in polynomial time using the seminal result of Arora, Rao and Vazirani~\cite{arora2009expander}. Feige et al.~\cite{feige2005improved} showed that in fact this guarantee can be made $\OO(\sqrt{\log k})$, and also showed that one can compute vertex separators with the same approximation ratio. In directed graphs, the \MB{} problem has been typically studied as \DBS: Given a directed graph $G$ with even number of vertices, is it possible to partition the vertex set into two equal parts $A$ and $B$ so that the number of arcs from $A$ to $B$ is at most $k$? This question was first raised by Feige and Yahalom~\cite{feige2003complexity}. When $k = 0$, they referred to this problem as \OC{}, and showed that even this problem is NP-hard. However, given a \OC{} instance which admits a solution, if one relaxes the requirement so that the algorithm can output a partition $(A',B')$ of $V(G)$ such that there are no arcs directed from $A'$ to $B'$ and $||A'| - |B'|| \leq \epsilon n$ where $\epsilon = \Omega(\frac{1}{\log n})$, the problem now becomes tractable, and they show a polynomial time algorithm for \OC. Madathil et al.~\cite{madathil2021sub} showed that the \DBS{} problem is FPT with respect to $k$, even when one requires $|A| = |B|$ exactly, on a subclass of directed graphs called semi-complete digraphs, which are the class of directed graphs where for every pair of vertices $u$ and $v$, there is an arc from $u$ to $v$ or an arc from $v$ to $u$. In terms of approximation algorithms, Agarwal et al.~\cite{agarwal2005log} showed an $\OO(\sqrt{\log n})$ approximation for \DB{} , the directed analogue of \BS{}, while Even et al.~\cite{even1999fast} showed an $\OO(\log k)$ approximation, which is the best known approximation guarantee depending only on $k$.

However, there is no prior work on the fixed-parameter tractability of \DB{} or \DBS{} for general $k$ on general directed graphs, even when we relax the requirement of finding a bisection to that of finding an \emph{approximate bisection}, that is, find a partition $(A', B')$ with $|A'|, |B'| = \Omega(n)$ so that the number of arcs from $A'$ to $B'$ is at most $k$.

Our result makes the first progress on this problem. Before we state our results, we define the \DB{} problem formally. Since all our results work for both the vertex and edge versions, we state them together as one problem. We adapt our definition from the definition of \BS{} in~\cite{fm06}, whose results we generalize. For completeness, we recall the definition of \BS{} in~\cite{fm06}.
\defparproblem{\BS}{Undirected graph $G = (V,E)$}{$k,b$}{Is there a set of vertices (edges) $F$ with $|F| \leq k$, so that in $G \setminus F$, every connected component has size at most $bn$?}

\defparproblem{\DB}{Directed graph $G = (V,E)$}{$k,b$}{Is there a set of vertices (arcs) $F$ with $|F| \leq k$, so that in $G \setminus F$, every strongly connected component has size at most $bn$?}

For the sake of clarity and comparison, we state the main result of~\cite{fm06}. Given an undirected graph $G$, we say that a set of vertices/edges $F$ is a \emph{$b$-balanced separator} if every connected component of $G \setminus F$ has size at most $bn$.

\begin{theorem}[\cite{fm06}]\label{thm:fm}
Given an instance of \BS{} with $\frac{2}{3} \leq b \leq 1$ there is a randomized algorithm, that for any $\epsilon > 0$, runs in time $2^{\OO\left(k\log \left(\frac{1}{\epsilon}\right)/{\epsilon^2}  \right)}(m + n)$ and with constant probability outputs either (a) a set of vertices (edges) $F'$ of size at most $k$ such that every connected component of $G \setminus F'$ has size at most $(b + \epsilon)n$ or (b) concludes correctly that there is no $b$-balanced separator of size at most $k$.
\end{theorem}
The following theorem, which directly generalizes the result of~\Cref{thm:fm} is our main result. Given a directed graph $G$, we say that a set of vertices/arcs $F$ is a \emph{$b$-balanced separator} if every \emph{strongly connected component} (SCC) of $G \setminus F$ has size at most $bn$.

\begin{restatable}{theorem}{balancedcut}\label{thm:main}
Given an instance of \DB{} there is a randomized algorithm, that for any $\epsilon > 0$, runs in time $2^{\OO\left(k \min\{\log\frac{1}{b}, \log k\}\log \frac{1}{\epsilon}/\epsilon^2 \right)}(m + n)$ and with constant probability outputs either (a) a set of vertices (arcs) $F'$ of size at most $k$ such that every strongly connected component of $G \setminus F'$ has size at most $(b + \epsilon)n$ or (b) concludes correctly that there is no $b$-balanced separator of size at most $k$.
\end{restatable}

We observe that our algorithm has a running time of $2^{\OO\left(k\log \left(\frac{1}{\epsilon}\right)/{\epsilon^2}  \right)}(m + n)$ for any $b = \Omega(1)$, matching the run-time of~\Cref{thm:fm} for $\frac{2}{3} \leq b \leq 1$ while also extending to any parameter $b \in (0,1)$. We also observe that~\Cref{thm:main} implies~\Cref{thm:fm}, since given any undirected graph $G$, one can create a directed graph $H$ on the same vertex set, so that for every edge $\{u,v\} \in E(G)$, we have the two arcs $(u,v), (v,u) \in E(H)$ and we can apply~\Cref{thm:main} to obtain~\Cref{thm:fm}.

Our algorithm for~\DB{} can be used to solve (approximate)~\DBS{} as well. To see this, given a graph $G$, observe that in any bisection $(A,B)$ so that the number of arcs going from $A$ to $B$ is at most $k$, the set $F$ of these at most $k$ arcs form a $\frac{1}{2}$-balanced separator, so that in $G \setminus F$, every strongly connected component is of size at most $\frac{n}{2}$. Using~\Cref{thm:main} with $b = \frac{1}{2}$, we can find a set of arcs $F'$ with $|F'| \leq k$ that forms a $(\frac{1}{2}+\epsilon)$ balanced separator. Finally, note that the strongly connected components of the graph $G \setminus F'$ form a Directed Acyclic Graph (DAG), and each strongly connected component of $G \setminus F'$ has size at most $(\frac{1}{2} + \epsilon)n$. It follows that there is a topological ordering $\{C_1, C_2 \ldots C_{\ell}\}$ of these strongly connected components of $G \setminus F'$, such that there is no arc from $C_j$ to $C_i$ for $i, j \in [\ell]$ with $j > i$. Therefore we can pick some prefix of strongly connected components in the topological ordering of $G 
\setminus F'$ to obtain a set $A'$, so that both $|A'|, |V(G) \setminus A'| \geq (1 - 2\epsilon)\frac{n}{4}$, and in $G \setminus F'$, there are no arcs from $V(G) \setminus A'$ to $A'$. It follows that the set of arcs $F'$ forms an (approximate) directed bisection. Note that the approximation is only in the balance, not in the number of arcs cut.

Next, we show a $\OO(\sqrt{\log k})$ approximation for \DB. This improves both the $\OO(\sqrt{\log n})$ approximation  of~\cite{agarwal2005log} and the $\OO(\log k)$ approximation given by Even et al.~\cite{even1999fast} for approximating \DB{} in polynomial time.

\begin{restatable}{theorem}{approxbalsep}
There is an $\OO(\sqrt{\log k})$ approximation to \DB{} in polynomial time. Formally, given an instance of \DB{} with $b = \Omega(1)$, suppose there is a set of vertices (arcs) $F$ with $|F| \leq k$, so that every strongly connected component of $G \setminus F$ has at most $bn$ vertices. Then there is a polynomial time randomized algorithm that with constant probability finds a set of vertices (arcs) $F'$ with $|F'| \leq \OO(k\sqrt{\log k})$ so that in $G \setminus F'$, every strongly connected component has size at most $b'n$ for some $b' < 1$ depending on $b$.
\label{thm:approxbalsep}
\end{restatable}

Note that for this theorem, we do not optimize $b'$, unlike our FPT result where we were able to show that $b' \leq b + \epsilon$ for some suitably chosen $\epsilon$. Still, we remark that is not a limitation of our framework, but inherent in~\cite{agarwal2005log} due to the use of the ARV separation theorem~\cite{arora2009expander}.

\subsection{Vertex Cut Sparsifiers}
Vertex sparsification is a fundamental problem in various settings. Broadly, given a graph and a terminal set $T$, a vertex sparsifier is a smaller graph (with size typically depending only on $|T|$) that preserves some (cut based) property of the terminals $T$. Moitra~\cite{moitra2009approximation} first introduced a version of vertex sparsification in undirected graphs. Chuzhoy~\cite{chuzhoy2012vertex} generalized this notion by allowing Steiner nodes in the sparsifier. Both these notions preserve edge cuts between bi-partitions of the terminal set.

For our setting, we focus on directed graphs and vertex cuts. Given a directed graph $G$ and a set of terminals $T$ along with an integer $c$, a $(c,T)$ vertex cut sparsifier for $G$ is another graph $G'$ with $T \subseteq V(G')$, so that for every partition $\{A, B\}$ of $T$ (so that $A \cup B = T$ and $A \cap B = \emptyset$), if the size of an $A$-$B$ vertex min-cut is at most $c$ in $G$, then the size of an $A$-$B$ vertex min-cut is the same in both in $G$ and $G'$. In other words, the goal is to find a vertex sparsifier that preserves vertex min-cuts of size at most $c$ between sets of terminals. Here we allow the deletion of terminals as well.

Kratsch and \Wahlstrom~\cite{kratsch2012representative,kratsch2020representative} first studied this notion of vertex sparsification for vertex cuts (without the parameter $c$) and showed that given $G,T$,  there is a randomized polynomial time algorithm that obtains a $(|T|, T)$ vertex sparsifier $G'$ for $G$ with $|V(G')| \leq \OO(|T|^3)$. This bound was improved to $\OO(|T|^2)$ by~\cite{he2021near} for the special case of directed acyclic graphs. Their algorithm runs in linear time for fixed $|T|$, but is still randomized as it needs to compute a representation of gammoids.   

One can then ask the question as to what is known about deterministic algorithms. Recently, Misra et al.~\cite{misra2020linear} showed that a representation of a gammoid with rank $r$ over a ground set of $m$ elements can be constructed in time $\OO({m \choose r} m^{\OO(1)})$ deterministic time. The technique of~\cite{kratsch2012representative} needs to compute a representation of a gammoid whose ground set has size $n$, the number of vertices and rank $|T|$, the number of terminals. We therefore obtain the following result.

\begin{theorem}[\cite{kratsch2012representative},~\cite{misra2020linear}]
Given a $n$-vertex graph $G$ and a terminal set $T$, there is a deterministic algorithm that runs in time $\OO({n \choose |T|} n^{\OO(1)})$ and computes a $(|T|, T)$ vertex cut sparsifier for $G$ of size at most $\OO(|T|^3)$.
\label{thm:detsparsifier}
\end{theorem}

However, this algorithm (which is an XP algorithm in the notation of parameterized complexity) can be easily improved to a (still deterministic) FPT algorithm. The reason for this is simple: there are at most $2^{|T|}$ partitions of the terminal set $T$. For each partition $A \cup B$, we can compute an $(A,B)$ directed min-vertex cut $M$. Note that this min-cut is of size at most $|T|$, since one can simply delete all the terminals to obtain a cut. This gives a set $X$ of $2^{|T|} \cdot |T|$ vertices. It can now be shown that we can apply the closure operation to the set $V(G) \setminus X$, where we simply delete all vertices of $V(G) \setminus X$, and for each vertex pair $(u,w)$ such that there are vertices $(v_1, v_2)$ with $v_1, v_2 \in V(G) \setminus X$ with $(u,v_1), (v_2, w) \in E(G)$ we add an arc $(u,w)$. This results in a sparsifier of size $\OO(|T|2^{|T|})$. One can now set $n = \OO(|T|2^{|T|})$ in~\Cref{thm:detsparsifier} to obtain a sparsifer of size $2^{\OO(|T|^2)}$ in time $2^{\OO(|T|^2)}$. The total running time is $2^{\OO(T)}(m + n) + 2^{\OO(|T|^2)}$.

A similar line of research considers edge and vertex cuts in undirected graphs. For edge cuts in undirected graphs, Chalermsook et al.~\cite{chalermsook2021vertex} showed an upper bound of $\OO(|T|c^4)$ which was later improved to $\OO(|T|c^3)$ by Liu~\cite{liu2023vertex}. Both these algorithms are randomized. Saranurak and Yingchareonthawornchai~\cite{saranurak2022deterministic} considered the vertex version in undirected graphs, and gave a deterministic algorithm to compute a sparsifier of size $\OO(|T|2^{\OO(c^2)})$ in time $\OO(m^{1 + o(1)}2^{\OO(c^2)})$. 

However, no improvement on the FPT algorithm with running time $2^{O(|T|)}(m + n) + 2^{O(|T|^2)}$ based on \Cref{thm:detsparsifier} is known for deterministic algorithms for vertex cut sparsifiers in directed graphs. We show the following result, which shows that if one is interested in preserving cuts of size at most $c$, then a better deterministic algorithm is possible.

\begin{restatable}{theorem}{mimicking}\label{thm:mimicking}
Given a directed graph $G$, a set of terminals $T$ and integer $c$, there is a deterministic algorithm that runs in time $\OO(\psi(|T|,c) \cdot (m + n))$ and computes a $(c,T)$ vertex sparsifier $G'$ for $G$ of size at most $\psi(|T|,c)$, where $\psi(|T|,c) = {|T| \choose \leq 3c} 2^{\OO(c)}$. Additionally, $G'$ satisfies $T \subseteq V(G') \subseteq V(G)$, with the property that for every partition $A \cup B$ of $T$ and every subset $X \subseteq V(G)$ with $|X| \leq c$, $X$ is an $A$-$B$ important separator in $G$ if and only if $X \subseteq V(G')$ and $X$ is an $A$-$B$ important separator in $G'$. 
\end{restatable}

While the dependence on $c$ in the exponent for the size of the sparsifier seems undesirable,~\Cref{thm:mimicking} shows that the vertex sparsifier $G'$ constructed by our algorithm is more powerful: it in-fact ``preserves'' all $A$-$B$ important separators for any partition $(A,B)$ of $T$. Additionally, if one wants only a vertex sparsifier, we note that simply running the algorithm in~\Cref{thm:detsparsifier} after applying our result in~\Cref{thm:mimicking} gives a vertex sparsifier of size $\OO(|T|^3)$, but now the running time is $\left(\frac{|T|}{c}\right )^{\OO(c|T|)}$. Thus, together, we obtain a vertex cut sparsifier of size $\OO(|T|^3)$ in time ${|T| \choose \leq 3c} 2^{\OO(c)}(m + n) + \left(\frac{|T|}{c}\right)^{\OO(c|T|)}$. 

\begin{corollary}
Given  a graph $G$ and a terminal set $T \subseteq V(G)$, there is a deterministic algorithm that runs in time ${|T| \choose \leq 3c} 2^{\OO(c)}(m + n) + \left(\frac{|T|}{c}\right)^{\OO(c|T|)}$ and obtains a $(c,T)$ vertex cut sparisifer for $G$ with $\OO(|T|^3)$ vertices.
\end{corollary}

\subsection{Techniques and Overview}

In this section we give a high-level overview of our techniques. For the sake of simplicity and clarity we sometimes omit small details in this section.

\paragraph{All-Subset Important separators:} Our result on counting all-subset important separators is obtained using the connection between important separators and \emph{closest sets}. For simplicity, in this overview, suppose that we are given a single source $s$ and a set of target vertices $T$. Our goal in this simplified setting shall be to bound the number of $(s,B)$ important separators of size exactly $k$ across all $B \subseteq T$, and we will show a weaker bound of $4^k{|T| \choose \leq {(k +1)^2}}$. Before we describe the main idea, we remark that our eventual goal will be to show that every $(s,B)$ important separator of size $k$ for some $B \subseteq T$ is an $(s,B')$ important separator for some $B' \subseteq B$ with $|B'| \leq (k +1)^2$. It is easy to see that this suffices, since there are only ${|T| \choose \leq (k + 1)^2}$ choices for $B'$, and there are at most $4^k$ such $(s,B')$ important separators of size $k$ for a fixed choice of $B'$.

Fix $B \subseteq T$, and consider an $(s,B)$ important separator $X$ of size $k$. Consider the (vertex) min-cut $X'$ between $X$ and $B$\footnote{Throughout this paper, we allow deletion of source and sink vertices in vertex cuts}. If this min-cut $X'$ has size strictly less than $|X| = k$, $X$ cannot be important: $X'$ would allow more vertices to be reachable from $s$ while having a size strictly less than $k$. This means that $X$ by itself must be a $(X,B)$ min-cut. Using simple cut-flow duality, it must be the case that there are $|X| = k$ vertex disjoint paths from $X$ to (some) vertices in $B$.

However, this flow property can be strengthened: it is easy to see that we can in fact assume that every such important separator $X$ is the $(X,B)$ min-cut ``closest'' to $B$. (We say $X$ is closest to $B$ if $X$ is the unique $(X, B)$ min-cut.)
But it can be shown that closest sets have a stronger flow property: we can now conclude that for every vertex $v \in X$, there are $k + 1$ paths from $X$ to $B$, which are vertex disjoint except that two of these paths both start at $v$ (\Cref{lemma:closest}). 

Why does this help? Fix a $v \in X$, and fix the $k + 1$ paths from $X$ to $B$. Suppose the $k + 1$ endpoints of these paths are $B^*_v \subseteq B$. Consider any set of vertices $X'$ that is a $(s,B^*_v)$ separator of size $\leq k$ that is ``closer'' to $B^*_v$ than $X$. (Here ``closer'' means that the set of vertices reachable from $s$ in $G \setminus X'$ is a superset of that in $G \setminus X$). Then $X'$ must delete $v$! For if not, since there are $k + 1$ paths from $X$ to $B^*_v$ which are vertex disjoint except $2$ paths which begin at $v$, and $|X'| \leq k$, $X'$ cannot be a $(s,B^*_v)$ separator. 

Applying the above argument for every $v \in X$ and letting $B' := \cup_{v \in X} B^*_v$ allows us to conclude that there is no $(s, B')$ separator $X'$ that is closer to $B'$ than $X$.
Thus, in fact, among all the sets of size at most $k$ separating $s$ from $B'$, $X$ is a set which is closest to $B'$. This means $X$ must be an important $(s,B')$ separator. But $|B'| = |\cup_{v \in X} B^*_v| \leq k(k + 1) \leq (k +1)^2$, and hence we are done. Our actual result uses a slightly more subtle argument using augmenting paths to obtain such a $B'$ with $|B'| \leq 2k$.

Here we remark that Lemma 4.10 of Lokshtanov et al.~\cite{lokshtanov2025wannabe} showed a similar result: it claims that there is a set $B' \subseteq B$ with $|B'| \leq k + 1$ such that any $(s,B)$ important separator is an $(s,B')$ important separator. However, there is a gap in the proof which leads to an unavoidable factor of $2$ loss. Indeed, as we show in~\Cref{lemma:impseppreserve}, there is a graph $G$, source vertex $s \in V(G)$ and $B \subseteq V(G)$ for which $|B'| = 2k$ is necessary, and hence our result is tight. Our high-level idea is flow-based, similar to~\cite{lokshtanov2025wannabe}. But as described above, the direct application of this flow-based idea gives a bound of $k(k + 1)$, which we refine using an augmenting path based argument to obtain the tight bound of $2k$.

\paragraph{Detection sets and Sample sets for Directed Graphs:}
To obtain detection sets and sample sets for directed graphs, on a high level, we follow the approach of Feige and Mahdian~\cite{fm06}.
We focus on sample sets in this brief overview. First, we consider the family of sets that arise as a strongly connected component after deleting an arbitrary $k$ vertices in the given digraph. Our key contribution is to show that this family has VC-dimension $\OO(k)$. Using the standard results on $\epsilon$-samples, similar to that in~\cite{fm06}, it then follows that a random set of $\OO(\frac{k}{\epsilon^2} \log \frac{1}{\epsilon})$ vertices is a sample set with constant probability. Thus the question becomes: how can we bound the VC-dimension of the family of sets $\SS$ formed by strongly connected components after deleting some $k$ vertices? We note that the techniques in~\cite{fm06} to bound the VC-dimension do not extend to directed graphs, and hence a different approach is needed. We accomplish this by showing a connection to our result on all-subset important separators. 

Recall the definition of VC-dimension: given a set family $\SS$ on a universe $U$, the VC-dimension is the size of the largest set $U' \subseteq U$ shattered by $\SS$. Here $U'$ is shattered by $\SS$ if for every subset $U'' \subseteq U'$ there exists an $S \in \SS$ with $|U' \cap S| = U''$. Thus to prove a bound of $\OO(k)$ for VC-dimension for our case, it is enough to show that there exists some constant $c$ so that any set of terminals $T \subseteq V(G)$ with $|T| \geq ck$ cannot be shattered by the set family $\SS$ which consists of the strongly connected components $C$ in $G \setminus F$ across all sets $F \subseteq V(G)$ with $|F| \leq k$. Fix any $T$ of size $\geq ck$ for large enough $c$ which we choose later. For the sake of contradiction, assume that the set $T$ can be shattered. Fix a ``pattern'' $P \subseteq T$, $P \neq \emptyset$. Since $T$ can be shattered by $\SS$, there exists a set $F$ of at most $k$ vertices, so that in $G \setminus F$, there is a strongly connected component $C$ that satisfies $C \cap T = P$. Fix some terminal $t \in P$. Then it must be the case that $P$ is the exact set of terminals that can both reach $t$ and can be reached from $t$ in $G \setminus F$. Equivalently, it is the exact set of terminals that can be reached from $t$ in both $G \setminus F$ and $G^R \setminus F$, where $G^R$ is obtained by reversing every arc of $G$. This motivates us to define $\Reach(t)$ for each $t \in T$, 
the collection of subsets of terminals that can be reached from $t$ after deleting a set $F$ of size at most $k$. Formally,
\begin{align*}
\Reach(t) = \{Q \subseteq T \mid\, & \text{there exists $F \subseteq V(G)$, $|F| \leq k$ such that $Q$ is reachable from $t$} \\&\text{while $T \setminus Q$ is unreachable from $t$ in $G \setminus F$}\}
\end{align*}
If we can bound $|\Reach(t)|$ for any graph by a function $B(|T|,k)$, then applying this result twice, once on $G$ and once on $G^R$, will bound the number of such patterns $P$ with $t \in P$ by $(B(|T|,k))^2$. Then a simple union bound over all $t \in T$ will give that the number of patterns $P \subseteq T$ is at most $|T|B(|T|,k))^2$. If this is less than $2^{|T|}$ whenever $|T| \geq ck$, then we have a contradiction to the fact that $T$ can be shattered. Thus the key question is to obtain $B(|T|,k)$, a good bound on $|\Reach(t)|$.

We do this as follows. Suppose $T' \in 
\Reach(t)$ some $T' \subseteq T$. Then there exists a set $F$ of at most $k$ vertices after whose deletion the set of reachable terminals from $t$ is exactly $T'$. We then show that there is a $(t, T \setminus T')$ \emph{important separator} $X$ of size at most $k$ whose deletion gives the same reachability set $T'$ as that after deleting $F$. But now using our bound on important separators, we can show that $B(|T|,k) \leq 4^k {|T| \choose \leq 2k}$. It is then easy to see that $|T|B(|T|,k)^2 < 2^{|T|}$ whenever $|T| \geq ck$ for some large enough (absolute) constant $c$.

\paragraph{FPT algorithm for Directed Balanced Separators:}
To obtain the algorithm for \DB{}, we follow the approach of~\cite{fm06} in the undirected case. For simplicity, in this overview, we work with vertex separators, assume that $b = \frac{1}{2}$, and we will only return a $b' = (\frac{3}{4} + 2\epsilon)$-balanced separator (as opposed to the $(b + \epsilon)$ balance guaranteed in our result). We are given that there is a set of vertices $F$ with $|F| \leq k$, so that every strongly connected component of $G \setminus F$ has at most $\frac{1}{2}n$ vertices. We first compute an $(\epsilon, k)$ sample set $T$ of size $\OO\left ( \frac{k \log \frac{1}{\epsilon}}{\epsilon^2}\right)$ using our result on sample sets,~\Cref{thm:sample_set}. Since $T$ is a sample set, each SCC of $G \setminus F$ has at most $(\frac{1}{2} + \epsilon)|T|$ terminals. We will try to find a set of vertices $F'$ such that $|F'| \leq k$ and every SCC of $G \setminus F'$ has at most $(\frac{3}{4} + \epsilon)|T|$ terminals. The property of sample sets will then again imply that $F'$ is a $(\frac{3}{4} + 2\epsilon)$-balanced separator.

Consider the toplogical ordering $C_1, C_2 \ldots C_{\ell}$ of the SCC's of $G \setminus F$, so that there is no arc from $C_j$ to $C_i$ for $i,j \in [\ell]$ with $j > i$. Consider the smallest index $i^*$ so that the union $L = \bigcup_{i = 1}^{i^*} C_i$ contains at least $(\frac{1}{4} - \epsilon)|T|$ terminals. Since no component contains more than $(\frac{1}{2} + \epsilon)|T|$ terminals, it follows that $L$ contains at most $\frac{3}{4}|T|$ terminals. Also, by definition, $V(G) \setminus L$ contains at most $(\frac{3}{4} + \epsilon)|T|$ terminals.

The algorithm proceeds as follows. First, guess $T \cap L$ and $T \cap (V(G) \setminus L)$. This can be done in time $2^{\OO\left(\frac{k \log \frac{1}{\epsilon}}{\epsilon^2}\right)}$ since $|T| = \OO\left(\frac{k\log \frac{1}{\epsilon}}{\epsilon^2}\right)$. Next, we compute a directed min-vertex cut between $T \cap (V(G) \setminus L)$ and $T \cap L$. Since $F$ is a vertex cut between these sets of size at most $k$, this min-cut $F'$ must be of size at most $k$ as well. But then the set of vertices $F'$ is a $(\frac{3}{4} + \epsilon)$-balanced separator with respect to the set $T$. Using the property of sample sets, it follows that $F'$ is a $(\frac{3}{4} + 2\epsilon)$-balanced separator in $G$.
 
 In order to make the loss in the balance factor only an additive $\OO(\epsilon)$ we need a slightly more sophisticated argument. We accomplish this using a reduction to the \SMC{} problem, where we are given pairs $(s_i, t_i), i \in [\ell]$, and the goal is to separate $s_i$ from all $t_j$, $j \leq i$, for each $i$, by deleting minimum number of vertices. This problem, which is a special case of {\textsc{Directed Multicut}}{}, was first defined by~\cite{chen2008fixed} in their FPT algorithm for \DFAS.

\paragraph{Approximation algorithm for Directed Balanced Separators:} Our $\OO(\sqrt{\log k})$ approximation algorithm proceeds similar to the FPT algorithm. The only key technical difference is that to compute a balanced separator with respect to the terminal set, we cannot guess where the terminals of the sample set are (this requires FPT time); instead we use the algorithm of~\cite{agarwal2005log} to obtain a balanced separator with respect to the sample set. While the algorithm of~\cite{agarwal2005log} has an approximation ratio of $\OO(\sqrt{\log n})$, when we need a balanced separator with respect to a terminal set $T$, this algorithm can be made to work with an approximation ratio of $\OO(\sqrt{\log |T|})$. On a high level, the reason is quite simple: the structure theorem of Arora, Rao and Vazirani \cite{arora2009expander} is a statement about vectors. As long as we apply the structure theorem to only $|T|$ vectors, we get an approximation ratio of $\OO(\sqrt{\log |T|})$.

\paragraph{Vertex Sparsifiers:} Our result on vertex sparsifiers is similar in spirit to the result by Kratsch and \Wahlstrom~\cite{kratsch2012representative} and proceeds by identifying \emph{irrelevant vertices}. However instead of using techniques based on matroids, we use our result on all-subsets important separators. Given a terminal set $T$, we are interested in preserving cuts across partitions of the terminal set of size up to $c$. We will define an irrelevant vertex as a vertex that is (a) not in $T$ and (b) not in any $A$-$B$ important separator of size $\leq c$ for any partition $A \cup B$ of $T$. By our result on all-subsets important separators, we can bound the number of relevant (non-irrelevant) vertices by ${|T| \choose  \leq 3c} 2^{\OO(c)}$, and identify this set in time ${|T| \choose \leq 3c} 2^{\OO(c)} (m + n)$. Thus at least $|V(G)| - |T| -  {|T| \choose \leq 3c} 2^{\OO(c)}$ vertices are irrelevant. Let $I$ be the set of irrelevant vertices. We now apply the standard operation of ``closing'' the irrelevant vertices $I$, which is to delete each $v \in I$, and for every pair $(w,y)$ where $w$ is an in-neighbour of $v_1$ and $y$ is an out-neighbour of $v_2$ for some $v_1,v_2 \in I$, to add the arc $(w,y)$. This operation is the vertex equivalent of edge contraction. The key observation here is that the closure operation \emph{can be applied to the entire set $I$ at once} as opposed to applying it to one vertex of $I$ and restarting the whole algorithm. This helps us achieve a linear-time algorithm. Thus we obtain another graph $G'$ with $V(G') \subseteq V(G)$, where $|V(G')| \leq |V(G) \setminus I| \leq  {|T| \choose \leq 3c} 2^{\OO(c)}$.

%% file: figure_contribution.tex
\resizebox{\textwidth}{!}{%
	\begin{tikzpicture}[
			node distance = 1cm and 0.5cm,
			mynode/.style={draw, align=left, minimum height=2em, minimum width=3em},
			mybignode/.style={ align=left, minimum height=2em, minimum width=3em},
			myarrow/.style={-Stealth},
			mylabel/.style={font=\scriptsize\itshape}
			]
			
			\node (combinatorics) [mybignode, font=\bfseries\Large, align=center, text width=5cm] {Combinatorics};
			
			\node (counting) [mynode, right=of combinatorics, text width=9cm] {\textbf{Counting important separators}\\ \emph{known}: between fixed source/sink vertices \\ \emph{this work}: between all subsets of terminals};
			\node (sample) [mynode, right=of counting, text width=8cm] {\textbf{Detection sets and sample sets}\\ \emph{known}: undirected graphs \\ \emph{this work}: directed graphs};
			\node (mimick) [mynode, below=of counting, text width=9cm] {\textbf{Deterministic directed vertex cut sparsifiers}\\ \emph{known}: $\OO(|T|^3)$ size, $2^{\OO(|T|^2)}(m + n)$ time\\ \emph{this work}: 1. $\OO({|T| \choose 3c})$ size, $\OO({|T| \choose 3c}(m + n))$ time \\2. $\OO(|T|^3)$ size, $(\frac{|T|}{3c})^{\OO(c|T|)}(m + n)$ time};
			\node (separator) [mynode, right=of mimick, text width=8cm] {\textbf{Directed balanced separators}\\ \emph{known}: \(O(\sqrt{\log n})\)-approximation\\ \emph{this work}: first FPT algorithm also running in linear time, \(O(\sqrt{\log k})\)-approximation};
			
			\node (algorithms) [mybignode,  left=of mimick, font=\bfseries\Large, align=center, text width=5cm]{Algorithms};
			
			\draw[myarrow] (counting) -- (sample);
			\draw[myarrow] (counting) -- (mimick);
			\draw[myarrow] (sample) -- (separator);
			
			\draw[dashed] ([xshift=1cm,yshift=-0.75cm]combinatorics.south west) -- ([xshift=19cm,yshift=-0.75cm]combinatorics.south east);

		\end{tikzpicture}
		}

%% file: result.tex
\section{All-Subsets Important Separators}
The goal of this section is to prove our result on all-subsets important separators.
\impsep*

As described in the introduction, our proof proceeds by capturing the relationship between important separators and closest sets. We then use well-known connectivity properties of closest sets to obtain our result.

\begin{definition}[Closest set]
For sets of vertices $X,T\subseteq V(G)$, $X$ is closest to $T$ if $X$ is the unique vertex mincut between $X$ and $T$.
\end{definition}

We use the following well-known fact about closest cuts.

\begin{lemma}[Lemma 18 of~\cite{he2021near}]\label{lemma:closest}
$X$ is closest to $T$ if and only if for each vertex $v\in X$, there are $|X|+1$ paths from $X$ to $T$ that are vertex-disjoint except at $v$, which appears in exactly two of these paths.
\end{lemma}

\begin{definition}[Minimal separators]
Given $X, S, T \subseteq V(G)$ such that $S \cap T = \emptyset$, $X$ is called an $S$-$T$ (inclusion-wise) minimal separator if in $G \setminus X$, there is no path from any $s \in S$ to any $t \in T$, and further, for every $v \in X$, there is a path between some $s \in S$ and $t \in T$ in $G \setminus (X \setminus \{v\})$.
\end{definition}

\begin{definition}
Given $X, S \subseteq V(G)$, the reachability set of the source vertices $S$ after removing $X$, that is, the set of all vertices $v$ reachable from some vertex of $S$ via a directed path after removing $X$, is denoted by $R^G_S(X)$.%
\end{definition}

\begin{definition}
Given $X, S, T \subseteq V(G)$ with $S \cap T = \emptyset$, a minimal $S$-$T$ separator $X$ is called an $S$-$T$ important separator if for every $S$-$T$ separator $X' \subseteq V(G)$ with $|X'| \leq |X|$, $R^G_S(X') \supset R^G_S(X)$ does not hold.
\end{definition}

We will drop the subscript or the superscript in $R^G_S(X)$ when the graph or the set of source vertices is clear from the context.

Informally, the idea of important separators is similar to that of a closest set: they are separators which cannot be pushed ``closer'' to the sink, without increasing the cut-size. The next lemma captures this relationship formally.

\begin{lemma}\label{lemma:equivalence}
Given $X, S, T \subseteq V(G)$, $X$ is an important $S$-$T$ separator if and only if (a) $X$ is a  minimal $S$-$T$ separator and (b) $X$ is closest to $T$.
\end{lemma}

\begin{proof}
Notice that by definition, any important $S$-$T$ separator must be a minimal $S$-$T$ separator. First, suppose to the contrary that $X$ is an important $S$-$T$ separator but $X$ is not closest to $T$. Let $Y$ be a vertex min-cut between $X$ and $T$ with $Y \neq X$. Since $X$ is not closest to $T$, such a $Y$ indeed does exist. Clearly, $|Y| \leq |X|$. We claim that $R(X) \subset R(Y)$, which will contradict that $X$ is an important separator. First, let us show that $R(X) \subseteq R(Y)$. Suppose for some vertex $v$, $v \in R(X)$ but $v \notin R(Y)$. Since $Y$ is a min-cut, there are vertex disjoint paths between each vertex $y \in Y$ and some vertex $t(y) \in T$ such that no internal vertex of these paths is from $X$. Consider a path $P$ from some vertex $s \in S$ to $v$ in $G \setminus X$. There must be such a path, since $v \in R(X)$. Let $y$ be a vertex of $Y$ on this path: there has to be such a vertex $y$ since $v \notin R(Y)$. Then $y$ is reachable from $s$ in $G \setminus X$. Also, there is a path from $y$ to $t(y)$ which does not contain any vertex of $X$. It follows that there is a path from $s$ to $t(y)$ in $G \setminus X$, which is a contradiction since $X$ is an $S$-$T$ separator. Since $X$ is a minimal $S$-$T$ separator and $Y \neq X$, it follows that there exists a $v \in X \cap R(Y)$, and hence $R(X) \subset R(Y)$.

Next, suppose that $X$ is both closest to $T$ and is also a minimal $S$-$T$ separator. Assume for the sake of contradiction that $X$ is not important. Since $X$ is not important, there is another separator $X'$ with $|X'| \leq |X|$ which is a minimal $S-T$ separator, and further $R(X) \subset R(X')$. Since $X' 
\neq X$, there exists a $v \in X \cap (V(G) \setminus X')$. By~\Cref{lemma:closest}, for every vertex $v \in X$ there exists a collection of $|X| + 1$ paths from $X$ to $T$ which are vertex disjoint except at $v$, which appears twice. Since $X'$ does not delete $v$, it follows that there is at least one such path $P$ from some vertex $x \in X$ to some vertex $t \in T$ which survives in $G \setminus X'$. We next show that there must be a path from some vertex of $S$ to $x$ in $G \setminus X'$ and this will imply a path from $S$ to $t$ in $G \setminus X'$, a contradiction.

Consider any path from some vertex $s \in S$ to $x$, all of whose internal vertices are not in $X$. Such a path must exist, for otherwise $X$ will not be a minimal $S$-$T$ separator since $X \setminus \{x\}$ will also be an $S$-$T$ separator. Every internal vertex on this path is reachable from $s$ in  $G \setminus X$. But $R(X) \subseteq R(X')$, which means every internal vertex on this path is reachable from $s$ in $G \setminus X'$ as well. This in turn means that there is a path from $s$ to $x$ in $G \setminus X'$, a contradiction.
\end{proof}

\begin{lemma}[Theorem 8.51 of~\cite{cygan2015parameterized}]
Given a digraph $G$ and two disjoint sets $A,B \subseteq V(G)$, there are at most $4^k$ $A$-$B$ important separators of size $\leq k$.
\end{lemma}

Equipped with these results, we are now ready to prove our main theorem in this section which bounds the number of important separators between across all subsets $A \subseteq S$ and $B \subseteq T$.

\impsep*

\begin{proof}
Let $X$ be an important $A$-$B$ separator for some $A \subseteq S$ and $B \subseteq T$ with $|X| \leq k$. We will show that $S$ is also an $A^*$-$B^*$ important separator for some $A^* \subseteq S$ and $B^* \subseteq T$ with $|A^*| \leq k$ and $|B^*| \leq 2k$.

By \Cref{lemma:equivalence}, it must be the case that $X$ is closest to $B$. Also, there are vertex disjoint paths from each vertex $x$ of $X$ to some vertex $t(x)$ of $B$. Also, we know that for every $v$ there exists $|X| + 1$ disjoint paths from $X$ to $B$ which are vertex disjoint except at $v$ which appears twice.

Consider the vertex-flow network $H$ obtained by adding a super-source $s^*$ and a super-sink $t^*$. We add outgoing arcs from $s^*$ to each vertex of $X$, and incoming arcs to $t^*$ from each vertex of $B$. The capacity of each vertex is $1$. Let $f$ be the (maximum) flow corresponding to any set of $|X|$ vertex disjoint paths between $X$ and $B$, and let the endpoints of these paths in $B$ be the set $Y$. Now we do the following operation for each vertex $v$. Consider the modified flow network $H_v$ obtained from $H$ where the underlying graph is same, but the vertex capacity of $v$ is raised to $2$. $f$ is not a maximum flow in $H_v$, since there exists $|X| + 1$ paths from $X$ to $B$ which are vertex disjoint except at $v$ which appears in two of these paths. Thus, there is an augmenting path for $f$ in $H_v$. We augment the flow $f$ along this augmenting path to obtain a new flow $f_v$. The new flow has one additional endpoint in $B$, so that the endpoints of the paths of this new flow $f_v$ are either in $Y$ or at the newly added vertex $u \in B \setminus Y$. Let us denote by $Y_v$ the set of the new endpoints. 

Let $B^* = \bigcup_{v \in X} Y_v$. Notice that $|B^*| \leq 2k$, since each $Y_v = Y \cup \{u\}$ for some vertex $u$ for each $v \in X$. Also since $X$ is minimal, there must exist paths $P_v$ for each $v \in X$ from some vertex $s(v)$ of $A$ to $v$ which do not contain any other vertex of $X$. Let $A^* = \bigcup_{v \in X} s(v)$. We now claim that $X$ is also an important $A^*-B^*$ separator.

First, observe that $X$ must be a minimal $A^*$-$B^*$ separator: there exists a path from some vertex of $A^*$ to each vertex $v \in X$ containing no other vertex of $X$, and similarly there exists a path from each vertex $v$ of $X$ to some vertex in $Y \subseteq B^*$ containing no other vertex of $X$. Now suppose that $X$ is not an $A^*$-$B^*$ important separator. Then there exists another $A^*$-$B^*$ minimal separator $X'$ with $|X'| \leq |X|$ with $R(X') \supset R(X)$. Since $X' \neq X$, there exists a vertex $v \in X \cap (V(G) \setminus X')$. Consider the set $Y_v$ and the set of $|X| + 1$ paths from $X$ to $Y_v$ which are vertex disjoint except at $v$, which appears twice. Since $X'$ does not delete $v$, at least one of these paths from some vertex $x \in X$ to $t \in Y_v$ survives in $G \setminus X'$. By construction, there is a path $P$ from some vertex $s \in A^*$ to $x \in X$, whose every internal vertex is not in $X$. Since every internal vertex of $P$ is in $R(X)$, it must also be in $R(X')$. It follows that there is a path from $s$ to $x$ in $G \setminus X'$. Appending this to the path from $x$ to $t$ we obtain a path from $s$ to $t$ in $G \setminus X'$, contradicting the fact that $X'$ is an $A^*$-$B^*$ separator.

Finally, it is clear that the number of $A^*$-$B^*$ important separators where $A^* \subseteq S$ and $B^* \subseteq T$ with $|A^*| \leq k$ and $|B^*| \leq 2k$ is at most $4^k$ times the number of 
ways of choosing $A^*$ and $B^*$,
and hence is at most $4^k {|S| \choose \leq k}{|T| \choose \leq 2k}$.
Furthermore, once we fix $A^*$ and $B^*$, all the important $A^*$-$B^*$ separators of size $\leq k$ can be enumerated in time $\mathcal{O}(4^k \cdot k^2 \cdot (m+n))$, see Theorem 8.51 of~\cite{cygan2015parameterized}. The set of all subsets of $A$ of size at most $k$ and the set of all subsets of $B$ of size at most $k$ can be enumerated in time $\OO(k {|S| \choose k})$ and $\OO(k {|T| \choose k})$ respectively. This completes the proof.
\end{proof}

\section{Detection Sets and Sample Sets in Directed Graphs}
The goal of this section is to prove our result on detection sets and sample sets in directed graphs. The following theorems are our main results in this section.
\detectionset*
\sampleset*

We start by defining a notion of $(\epsilon,k)$ nets and recalling the definitions for $(\epsilon, k)$ detection sets and  $(\epsilon,k)$ sample sets from the overview. These definitions naturally extend those in~\cite{fm06} for undirected graphs.

\begin{definition}[Net]
Given a directed graph $G$, an $(\epsilon, k)$ net is a set of terminals $T \subseteq V(G)$ satisfying the following property: for every set of vertices $F$ with $|F| \leq k$, the following two conditions are met:
\begin{enumerate}
\item For every SCC $C$ of $G \setminus F$ with $|C| \geq \epsilon n$ we have 
$$|C \cap T| \geq 1.$$
\item Let $C^*$ be the union of all except one SCC in $G \setminus F$ such that $|C^*| \geq \epsilon n$, then we have
$$|C^* \cap T| \geq 1.$$
\end{enumerate}
\end{definition}

\begin{definition}[Detection Set]\label{def:detection_set}
Given a directed graph $G$, an $(\epsilon, k)$ detection set is a set of terminals $T \subseteq V(G)$ satisfying the following property: for every set of vertices $F$ with $|F| \leq k$ such that $V(G \setminus F)$ can be partitioned into $(A,B)$ with $|A|, |B| \geq \epsilon n$ and there are no arcs from $B$ to $A$, there exists $t_1, t_2 \in T$ such that there is no $t_1$-$t_2$ path in $G \setminus F$.
\end{definition}

\begin{definition}[Sample Set]\label{def:sample_set}
Given a directed graph $G$, an $(\epsilon, k)$ sample set is a set of terminals $T \subseteq V(G)$ satisfying the following property: for every set of vertices $F$ with $|F| \leq k$, and every SCC $C$ of $G \setminus F$, we have 
$$\left |\frac{|C \cap T|}{|T|} - \frac{|C|}{n}\right | \leq \epsilon.$$

\end{definition}
We note that these definitions are almost identical to the one in~\cite{fm06}, with the only difference being that we deal with directed graphs and strongly connected components in lieu of undirected graphs and connected components. Before proving the main results of this section, we show that $(\epsilon,k)$-nets are also $(\epsilon,k)$-detection sets. The proof is similar to that in undirected graphs.

\begin{lemma}
Every $(\epsilon, k)$-net is also an $(\epsilon,k)$ detection set.

\end{lemma}

\begin{proof}
Let $T$ be an $(\epsilon,k)$-net for $G$.
    Consider a set of vertices $F$ of size at most $k$, so that $V(G \setminus F)$ admits a partition  $A,B$ with $|A|,|B| \geq \epsilon n$ such that there are no arcs from $B$ to $A$. Consider an SCC $C$ of $G \setminus F$, and let $C'$ be the union of all SCC's of $G \setminus F$ except $C$. Then clearly $|C'| \geq \epsilon n$, hence there exists a terminal $t_1 \in C' \cap T$. Let $C_{t_1}$ be the SCC of $G \setminus F$ containing $t_1$. Now consider $C''$ which is the union of all SCC's of $G \setminus F$ except $C_{t_1}$. Again $|C''| \geq \epsilon n$, therefore there is a terminal $t_2 \in  C'' \cap T$. But clearly $t_1$ and $t_2$ lie in different SCC's of $G \setminus F$, which means there is either no $t_1$-$t_2$ path or no $t_2$-$t_1$ path in $G \setminus F$, hence satisfying our detection set property.
\end{proof}

Thus we henceforth focus on obtaining $(\epsilon,k)$-nets and samples.

First, we observe that showing a VC-dimension bound directly implies $(\epsilon,k)$ samples. This is similar to the approach in~\cite{fm06}. We start with a few definitions.

\begin{definition}[Shattering a set of elements] Suppose we are given a set system $(\SS, U)$ consisting of a family of sets $\SS$, where each set in $\SS$ consists of elements from a universe $U$. We say that a subset $W \subseteq U$ is \emph{shattered} by $\SS$, if for every subset $Y \subseteq W$, there exists a set $S \in \SS$ so that $S \cap W = Y$.

\begin{definition}[VC-dimension]
Given a set system $(\SS, U)$, its VC-dimension is the size of the largest subset $W \subseteq U$ that can be shattered by $\SS$.    
\end{definition}

\end{definition}

\begin{theorem}[$\epsilon$-net theorem, see~\cite{fm06}]\label{thm:eps_net}
Let $(\SS, U)$ be a set system with VC-dimension $d$, and universe size $|U| = n$. Then for every $\epsilon > 0$, there exists an absolute constant $c$ such that a random subset $T \subseteq U$ of size $c{d}\frac{1}{\epsilon} \log \frac{1}{\epsilon}$ is  an $\epsilon$-net with probability at least $\frac{2}{3}$.
Concretely, $T$ satisfies $|S \cap T| \geq 1$ for every $S \in \SS$ satisfying $|S| \geq \epsilon n$ with probability at least $\frac{2}{3}$.
\end{theorem}

\begin{theorem}[$\epsilon$-sample theorem, see~\cite{fm06}]\label{thm:eps_sample}
Let $(\SS, U)$ be a set system with VC-dimension $d$, and universe size $|U| = n$. Then for every $\epsilon > 0$, there exists an absolute constant $c$ such that a random subset $T \subseteq U$ of size ${cd}\frac{1}{\epsilon^2} \log \frac{1}{\epsilon}$ is an $\epsilon$-sample with probability at least $\frac{2}{3}$. Concretely, $T$ satisfies $\left|\frac{|S|}{n} - \frac{|S \cap T|}{|T|}\right| \leq \epsilon$ for every $S \in \SS$ with probability at least $\frac{2}{3}$.
\end{theorem}

Let $\SS$ be the family of sets consisting of all possible sets $C$ which are either (a) a strongly connected component after deleting some vertex set $F$ of size at most $k$ or (b) the union of all but one strongly connected components after deleting some vertex set $F$ of size at most $k$. Then it suffices to prove a VC-dimension bound of $\OO(k)$ for this family of sets with the universe as the vertex set - this would then directly imply~\Cref{thm:detection_set,thm:sample_set} using~\Cref{thm:eps_net,thm:eps_sample}.
Thus we will henceforth focus on upper-bounding the VC-dimension of the set system $\SS$.

Towards this, we consider a slightly different problem of reachability from a single source after deleting $k$ vertices.

\begin{definition}[Single source reachability profile]
Given a graph $G$, a source vertex $s$ and a set of sink vertices $T'$ with $s \notin T'$, we define the single source reachability profile of $s$ as $\Reach^k(s,T') = \{ R^G_{ \{ s \} } (F) \cap T' \mid F \subseteq V(G), |F| \leq k \}$ 
In other words, $\Reach^k(s,T')$ is the collection of subsets of $T'$ reachable from $s$ after deleting some set of at most $k$ vertices from $G$.

\end{definition}

The next theorem bounds the size of $\Reach^k(s, T')$. 

\begin{theorem}\label{thm:reachability}
Given a graph $G$, source vertex $s$ and sink vertices $T'$, $|\Reach^k(s, T')| \leq 4^k {|T'| \choose \leq 2k}$.
\end{theorem}

We prove~\Cref{thm:reachability} using our result on all-subset important separators,~\Cref{thm:importantbound}. We do this as follows. Given $s$ and $T'$, suppose there exists a set $X$ of $k$ vertices so that $P \subseteq T'$ is the subset of $T'$ reachable from $s$ in $G \setminus X$. Then in the following lemma, \Cref{lemma:wlgclosest}, we show that there is in fact an $s$-$(T' \setminus P)$ \emph{important separator} $X'$ of size $\leq k$, so that the subset of $T'$ reachable from $s$ in $G \setminus X$ remains $P$. The bound on all-subset important separators will then imply the bound on the reachability profile.  

\begin{restatable}{lemma}{wlgclosest}\label{lemma:wlgclosest}
Suppose there exists a set $X$ of $\leq k$ vertices so that $P \subseteq T'$ is the subset of $T'$ reachable from $s$ in $G \setminus X$. Then there is an $s$-$(T' \setminus P)$ important separator $X'$ of size $\leq k$, so that the subset of $T'$ reachable from $s$ in $G \setminus X$ remains $P$.
\end{restatable}
\lv{
\begin{proof}
Without loss of generality, we assume that $X$ is a minimal $s-(T'\setminus P)$ separator (otherwise we consider the subset of $X$ that is minimal, this still has the same reachability set $P$). Let $X'$ be the vertex mincut between $X$ and $T'\setminus P$ that is closest to $T'\setminus P$. We have $|X'|\leq |X|$ since $X$ itself is a vertex cut between $X$ and $T'\setminus P$.

We show that the reachability profile after removing $X'$ is also $P$. In other words, a vertex $t\in T'$ is reachable from $s$ after removing $X$ iff it is reachable after removing $X'$.

Consider some sink vertex $t \in T' \setminus P$, so that $t$ is not reachable after removing $X$. Then, every $s$--$t$ path contains a vertex in $X$. Since $X'$ is a cut between $X$ and $T' \setminus P$,
this path must also contain a vertex in $X'$. This proves one direction.

Suppose that some vertex $t \in T'$ is not reachable after removing $X'$. Then, every $s$--$t$ path contains a vertex in $X'$. Suppose for contradiction that some $s$--$t$ path does not contain a vertex in $X$. Consider the prefix of the path leading to a vertex $v$ in $X'$. Append to it a path from $v$ to $T' \setminus P$ that only contains one vertex in $X'$ (at the start) and no vertex of $X$. Such a path must exist since $X'$ is a min-cut between $X$ and $T' \setminus P$. The result is a path from $s$ to $T\setminus P$ that does not contain a vertex in $X$. This contradicts the set $P$ being the reachability profile after removing $X$.

Finally, we show that $X'$ is an important $s$-$(T' \setminus P)$ separator. By~\Cref{lemma:equivalence} it is enough to show that $X'$ is closest to $T' \setminus P$ and that $X'$ is a minimal $s$-$(T' \setminus P)$ separator. It follows from the construction of $X'$ that $X'$ must be closest to $T' \setminus P$. Thus it is enough to show that $X'$ is a minimal $s$-$(T' \setminus P)$ separator. Since $X'$ is a min-cut between $X$ and $T' \setminus P$, for every vertex $x' \in X'$, there exists a path $Q$ from some vertex $x \in X$ to some vertex $t \in T' \setminus P$ through $x'$ that does not contain any other vertex of $X'$. Since $X$ is a minimal $s$-$(T'\setminus P)$ separator, there is a path $P'$ from $s$ to each vertex $x \in X$ that does not include any other vertex from $X$. Further $P'$ does not contain any vertex from $X'$: if it contains some vertex $y' \in X'$, it must be the case that $y' \notin X$. Then it follows that there is a path from $s$ to $y'$ in $G \setminus X$. Since $X'$ is a min-cut between $X$ and $T' \setminus P$, there is a path from each vertex of $X'$ to some vertex of $T'\setminus P$ that does not have any internal vertex from $X$.   
In particular, there is a path from $y'$ to some vertex of $T' \setminus P$ which does not include any vertex of $X$. This path,  when appended to the path from $s$ to $y'$ in $G \setminus X$, results in a path from $s$ to some vertex of $T' \setminus P$ which does not contain any vertex of $X$, which in turn contradicts the fact that $X$ is a cut between $s$ and $T' \setminus P$. Finally, appending $P'$ with $Q$ gives a path from $s$ to $t \in T' \setminus P$ that contains $x'$, while containing no other vertex of $X'$. Since such a path exists for any $x' \in X'$, it follows that $X'$ is minimal, and hence $X'$ is an important $s$-$T' \setminus P$ separator.
\end{proof}
}

\begin{proof}[Proof of~\Cref{thm:reachability}]
Suppose $P \in \Reach^k(s,T')$.~\Cref{lemma:wlgclosest} implies that there is an important $s$-$(T'\setminus P)$ separator $X$ so that the reachability set of $s$ in $G \setminus X$ is exactly $P \subseteq T'$. Using~\Cref{thm:importantbound} the number of such important separators is at most $4^k{|T'| \choose \leq 2k}$, and hence the number of such $P$ is at most $4^k{|T'| \choose \leq 2k}$.
\end{proof}

Next, we show why~\Cref{thm:reachability} implies a VC-dimension bound.

\begin{theorem}
Let $\SS$ be the family of sets over consisting of all possible sets $C$ which are either (a) a strongly connected component after deleting some vertex set $F$ of size at most $k$ or (b) the union of all but one strongly connected components after deleting some vertex set $F$ of size at most $k$. Then the VC-dimension of $\SS$ is $\OO(k)$.
\end{theorem}

\begin{proof}
From the definition of VC-dimension, it is enough to show that there no set of size $\geq dk$ that can be shattered, for some constant $d$. Suppose a set of terminals $T$ can be shattered. We will show that $|T| = \OO(k)$. 

Consider some $P \subseteq T$, with $P \neq \emptyset$, and suppose that after deleting some subset $F$ of vertices of size at most $k$, some strongly connected component $C$ of $G \setminus F$ satisfies $C \cap T = P$. Fix any $s \in C \cap T$, and let $T' = T \setminus \{s\}$. Let $G^R$ denote the graph $G$ with arcs reversed. Then by definition of an SCC, $P \setminus s$ is the exact subset of $T'$ reachable from $s$ in both $G \setminus F$ and $G^R \setminus F$. But this means $P \setminus s = P_1 \cap P_2$ for some $P_1 \in \Reach_G^k(s,T')$ and $P_2 \in \Reach_{G^R}^k(s, T')$.

Together, we obtain that if $P \subseteq T$ is $C \cap T$ for some SCC $C$ after deleting some $k$ vertices, there must exist a vertex $s \in P$ such that $P \setminus \{s\} = P_1 \cap P_2$ for some $P_1 \in \Reach_G^k(s,T \setminus \{s\})$ and $P_2 \in \Reach_{G^R}^k(s, T \setminus \{s\})$. But then the number of choices for $P$ is at most $|T|(4^k {|T| \choose \leq 2k})^2$.

Alternatively, suppose that $P = C' \cap T$ where $C'$ is the union of all but one strongly connected component, say $C^*$, after deleting some set $F$ of $k$ vertices. Suppose $P \neq T$, so that $C^* \cap T \neq \emptyset$. Consider some $s' \in C^* \cap T$. Then it follows that $(T \setminus P) \setminus \{s'\} = Q_1 \cap Q_2$ where $Q_1 \in \Reach^k_{G}(s', T \setminus \{s'\})$ and $Q_2 \in \Reach_{G^R}^k(s', T \setminus \{s'\})$. But then the number of choices for $T \setminus P$ is at most $|T|(4^k {|T| \choose \leq 2k})^2$, and hence the number of choices for $P$ is at most $|T|(4^k {|T| \choose \leq 2k})^2$.

Thus the total number of choices for $P$ across both cases is at most $2|T|(4^k {|T| \choose \leq 2k})^2 + 2$ (accounting for $P = \emptyset$ in the first case, and $P = T$ in the second)
which is not less than $2^{|T|}$ only when $|T| = \OO(k)$, finishing the proof.
\end{proof}

%% file: technical.tex
\lv
{
\section{FPT Algorithm for Finding Directed Balanced Separators}

The goal of this section is to prove our result on \DB.

\balancedcut*

We will crucially exploit our results on sample sets,~\Cref{thm:sample_set}. We will also need the following result, which is about computing skew separators. The \SMC{} problem is a special case of {\textsc{Directed Multicut}}.

\begin{definition}[\SMC{}~\cite{chen2008fixed}] Given a directed graph $G$ and a set of $\ell$ terminal pairs $\{ (s_i, t_i) \}_{i \in [\ell]}$ and an integer $k$, find a set of vertices $F$ of size at most $k$, so that $G \setminus F$ has no $s_i - t_j$ path for $i \geq j$, $i,j \in [\ell]$.\footnote{Again, we note that we can delete terminals as well to form the solution.}   
\end{definition}
The analagous edge version, \SMCE{} can be defined similarly.

\begin{theorem}[Extension of Theorem 8.41 of~\cite{cygan2015parameterized}]\label{thm:reduction}
\SMC{} and \SMCE{} admit FPT algorithms running in time $\OO(4^k \cdot k^3 \cdot (n + m))$.
\end{theorem}

\begin{proof}
Theorem 8.41 of~\cite{cygan2015parameterized} is for the edge version \SMCE. For the sake of completeness, we reduce the vertex version to the edge version. The reduction is standard and as follows.

Given an instance of \SMC{}, we create an instance of \SMCE{}. Define graph $G'$ as follows. For each vertex $v \in V(G)$, create two vertices $v_{in}$ and $v_{out}$ in $V(G')$. For every edge $\{u,v\} \in E(G)$, add an edge between $u_{out}$ and $v_{in}$ in $E(G
')$. Finally, add an edge $u_{in}$ to $u_{out}$ to $E(G')$ for each $u \in V(G)$. The instance of \SMCE{} is then the graph $G'$, the integer $k$ and the set of pairs $(s_{i_{in}}, t_{i_{out}})$ for each $i \in [\ell]$.

We claim that there for any integer $k'$, there is a solution to \SMC{} of size at most $k'$ in $G$ iff there is a solution to \SMCE{} of size at most $k'$ in $G'$. First we prove the forward direction. Let $X$ be a set of vertices that forms a solution to the \SMC{} instance. Then clearly by our construction, the set of edges $(x_{in}, x_{out})$ for $x \in X$ forms a solution to the $\SMCE{}$ instance.

For the backward direction, let $X'$ be a set of edges that forms a solution in $G'$. If an edge in $X'$ is of the form $(y_{out}, x_{in})$, we replace it by the edge $(y_{out}, y_{in})$, so that now every edge in $X'$ is of the form $(y_{out}, y_{in})$. Then let $Y = \{y \mid (y_{in}, y_{out}) \in X'\}$. Clearly, $|Y| \leq k$, and $Y$ must form a solution in $G$ for the \SMC{} instance. 
\end{proof}

\begin{lemma}\label{lemma:aux}
Let $T$ be an $(\epsilon,k)$ sample set. Suppose that the \DB{} instance is a YES-instance, so that there exists a set of vertices (arcs) $F$ with $|F| \leq k$ whose deletion leaves every strongly connected component with at most $bn$ vertices. 
Then one can in $\OO(2^{\OO(|T|\min\{\log \frac{1}{b}, \log |T|\})}m)$ time, find a set of vertices (arcs) $F'$ with $|F'| \leq k$ so that there is a partition $(X,Y)$ of $G \setminus F'$ such that every SCC of $G \setminus F'$ has size at most $(b + \OO(\epsilon))n$.
\end{lemma}

\begin{proof}
Broadly, our goal is to obtain a set of vertices (arcs) $F'$ with $|F'| \leq k$, so that in $G \setminus F'$, every strongly connected component does not have more than $(b + \epsilon)|T|$ terminals. The sample set property would then imply that every strongly connected component of $G \setminus F'$ would have size at most $(b + 2\epsilon)n$.

We proceed as follows. Let $C$ be a strongly connected component of $G \setminus F$. Then by the sample set property, we must have $|\frac{n}{|T|}|C \cap T| - |C|| \leq \epsilon n$, which in turn means that $|C \cap T| \leq (b + \epsilon)|T|$. Consider the strongly connected components of $G \setminus F$. They form a DAG. We sort this DAG in the topological order - let $(C_1, C_2 \ldots C_l)$ be the ordering of strongly connected components such that there is no arc from $C_j$ to $C_i$ whenever $j > i$, for $i, j \in [\ell]$.

\begin{itemize}

\item \textbf{Case 1:} We assume that $\log \frac{1}{b} \leq \log |T|$. Partition the indices $\{1, 2 \ldots \ell\}$ into $t \leq \frac{1}{b}$ intervals $[l_i, r_i]^t_{i = 1}$ as follows. Let $l_1 := 1$. For each $i$, $r_i$ is defined as the smallest integer greater than or equal to $l_i$ for which $\bigcup_{j = l_i}^{r_i} C_j$ 
contains at least $b|T|$ terminals. If no such integer exists, then $r_i := \ell$ and we set $t:= i$. If $r_i \neq \ell$ we set $l_{i+1} := r_i + 1$ and continue.  

For each terminal $t \in T$, we guess if $t \in F$, and if not, we guess $i$ such that $t \in \bigcup_{j = l_i}^{r_i} C_j$. We also guess if $t \in C_{r_i}$. All these guesses together take time $2^{\OO(|T|\log\frac{1}{b})}$.

We now create a \SMC{} (\SMCE{}) instance as follows. Let $T_{i1} = ((\bigcup_{j = l_i}^{r_i} C_j) \setminus C_{r_i}) \cap T$ and $T_{i2} = C_{r_i} \cap T$ for each $i \in [\ell]$. Observe that by construction, for any $i \in [\ell]$, both $T_{i1}$ and $T_{i2}$ have at most $(b + \epsilon)|T|$ terminals. Next, contract  $T_{i1}$ into a single vertex $t_{i1}$, and contract  $T_{i2}$ into a single vertex $t_{i2}$. Let $G^C$ denote the contracted graph. The instance of \SMC{} is the graph $G^C$ and the pairs $\{(t_{11}, t_{12}), (t_{12}, t_{21}), (t_{21}, t_{22}) \ldots (t_{\ell 1}, t_{\ell 2})\}$. (If some $T_{i1} = \emptyset$, we replace the pairs $(t_{{i-1},2}, t_{i1}), (t_{i1}, t_{i2})$ by the single pair $(t_{{i-1},2}, t_{i2})$, we omit this detail.)

\SMC{} (\SMCE{}) then returns a set of vertices (arcs) $F'$ with $|F'| \leq k$, so that in $G \setminus F'$, there is no path from $t_{c_1d_1} \rightarrow t_{c_2d_2}$ whenever $c_2 \leq c_1$ or when $c_2 = c_1$ and $d_2 \leq d_1$, for $c_1,c_2 \in [\ell]$ and $d_1,d_2 \in \{1,2\}$. In particular, this means that no strongly connected component of $G \setminus F'$ has more than $(b + \epsilon)|T|$ terminals, since for any $c \in [\ell]$ and $d \in \{1,2\}$, $|T_{cd}| \leq (b + \epsilon)|T|$.

\item \textbf{Case 2:} We now consider the case $\log \frac{1}{b} \geq \log |T|$, so that $|T| < \frac{1}{b}$. Observe that in the above analysis, the number of indices $i$ satisfying $C_i \cap T \neq 0$ is at most $|T|$. We guess the ordered partition $(W_1, W_2 \ldots W_{p})$ of $T \setminus F'$, where $p \leq |T|$, such that for each $j \in [p]$, $W_j = C_i \cap T$ where $C_i$ is the $j^{th}$ component in $(C_1, C_2 \ldots C_{\ell})$ satisfying $C_i \cap T \neq \emptyset$. In other words, we guess the ordered
partition of the terminal set induced by the strongly connected components $(C_1, C_2 \ldots C_{\ell})$ of $G \setminus F$. This can be accomplished in time $|T|^{|T|} = 2^{|T| \log |T|}$. 

We now create a \SMC{} (\SMCE{}) instance as follows. For each $j \in [p]$, contract  $W_{j}$ into a single vertex $w_{j}$. Let $H^C$ denote the contracted graph. The instance of \SMC{} is the graph $H^C$ and the pairs $\{(w_{1}, w_{2}), (w_{2}, w_{3}), (w_{3}, w_{4}) \ldots (w_{p-1}, w_{p})\}$. 

\SMC{} (\SMCE{}) then returns a set of vertices (arcs) $F'$ with $|F'| \leq k$, so that in $G \setminus F'$, there is no path from $w_{c_1} \rightarrow w_{c_2}$ whenever $c_2 \leq c_1$. In particular, this means that no strongly connected component of $G \setminus F'$ has more than $(b + \epsilon)|T|$ terminals, since for any $c \in [p]$, $|W_{c}| \leq (b + \epsilon)|T|$. 

\end{itemize}
Finally, by the property of $(\epsilon, k)$ sample sets (\Cref{def:sample_set}), in either case, it follows that since any strongly connected component of $G \setminus F'$ has at most $(b + \epsilon)|T|$ terminals, it must have at most $(b + 2\epsilon)n$ vertices. Note that though the statement of~\Cref{def:sample_set} 
 only deals with the case when $F'$ is a set of vertices, we have the simple observation that if $F'$ is a set of arcs and $C$ is a strongly connected component after deleting $F'$, it is also a strongly connected component after deleting the endpoints of these arcs not in $C$. This completes the proof.

\end{proof}

The proof of~\Cref{thm:main} now follows immediately from~\Cref{thm:sample_set} and~\Cref{lemma:aux}.

\section{Approximation Algorithm for Directed Balanced Separators}

In this section, the goal is to prove the following result, which shows a $\OO(\sqrt{\log k})$ approximation for \DB{}. 
\approxbalsep

Our result is obtained by essentially following the algorithm of~\cite{agarwal2005log} together with our theorem on sample sets,~\Cref{thm:eps_sample}. On a high level, the reason we get $\OO(\sqrt{\log k})$ approximation comes down to the fact that one can find a balanced separator with respect to a sample set $T$, which will automatically be a balanced separator for the entire graph as well due to the property guaranteed by sample sets. Once we have this equivalence, in order to obtain a balanced separator with respect to the set $T$, we now need to use the ARV structure theorem~\cite{arora2009expander} only on the sample set vectors, which are at most $|T|$ in number. By~\Cref{thm:sample_set} there is such a sample set $T$ of size at most $\OO(k)$. This helps replace the $\OO(\sqrt{\log n})$ factor by $\OO(\sqrt{\log k})$. However, for the sake of completeness and clarity we give the full algorithm by mostly following the algorithm of~\cite{agarwal2005log} with a few modifications/simplifications.

We will restrict our attention to edge cuts, as again the standard reduction in~\Cref{thm:reduction} can easily reduce the vertex version to the edge version.

Notice that if a separator $F$ is $b$-balanced for some $b = \Omega(1)$, then every strongly connected component of $G \setminus F$ has size at most $bn$. This in turn means that using a prefix of the topological order of the strongly connected components in $G \setminus F$, we can obtain sets $A,B$ with $|A|, |B| \geq c'n$ for  $c' = \frac{1-b}{2}$, such that there are no arcs from $A$ to $B$ in $G \setminus F$. Thus, upto constant factors in the balance, one can equivalently think of the~\DB{} problem as that of finding a set of vertices (arcs) $F$ with $|F| \leq k$ and a partition $A,B$ of $G \setminus F$, so that $|A|, |B| \geq c'n$ and there are no arcs from $A$ to $B$ in $G \setminus F$. 

Now we consider the terminal version of the problem, to which we will reduce via sample sets later. Given a terminal set $T$, our revised goal is to solve the following problem: given a parameter $c = \Omega(1)$, find a set of edges $F$ with $|F| \leq k$ so that there is a partition of $G \setminus F$ into $A, B$ such that $|A \cap T|, |B \cap T| \geq c|T|$ and there are no arcs from $A$ to $B$ in $G \setminus F$. In what follows, we will mostly follow the algorithm of~\cite{agarwal2005log} with small modifications which we state as and when required.

We assume that the vertices are labelled $\{1,2, \ldots n\}$, and add a special vertex $0$, which will be a non-terminal and will be the reference vertex for the 
``$A$- side'' of the cut. We start with the SDP relaxation (extended in our case to the terminal version) in~\cite{agarwal2005log} and follow their algorithm.

\begin{center}
\noindent\fbox{

  \begin{minipage}{0.95\textwidth}
  $$\min \frac{1}{8} \sum_{e=\{i,j\} \in E(G)} d(i,j)$$
  $$\|v_i\|^2 = 1 \;\;\forall i \in [n] \cup \{0\}$$
  $$\|v_i - v_j\|^2 + \|v_j - v_k\|^2 \geq \|v_i - v_k\|^2 \;\;\forall i,j,k \in [n] \cup \{0\}$$
  $$\sum_{i < j, i,j \in T} \|v_i - v_j\|^2 \geq 4c(1-c)|T|^2$$

  \end{minipage}
 
      }
\end{center}
Here the ``directed distance'' $d(i,j) = \|v_i - v_j\|^2 +\|v_j - v_0\|^2 - \|v_i - v_0\|^2$ is as defined in~\cite{agarwal2005log}. The cannonical solution for this SDP is obtained by setting $v_i = v_0$ for each $i \in A$ and $v_i = -v_0$ for each $i \in B$.

\begin{theorem}[ARV separation theorem~\cite{arora2009expander}, as stated in~\cite{agarwal2005log}] Given a set of $\ell_2^2$ unit vectors $v_i$, $i \in [n]$, such that $\sum_{i < j} \|v_i -v_j\|^2 \geq 4c(1-c)n^2$, there exists a polynomial time algorithm that finds disjoint sets $L,R$ with $|L|, |R| = \Omega(n)$ such that for any $i \in L, j \in R$, we have $\|v_i - v_j\|^2 \geq \Omega({1}/{\sqrt{\log n}})$.
\label{thm:structure}

\end{theorem}

Now we follow Algorithm 4 of~\cite{agarwal2005log}. Step 1 solves the SDP. In Step 2, given the vectors from the SDP, we apply the ARV separation theorem to only the vectors $v_i$, $i \in T$ to find $\Delta = \frac{1}{\sqrt{\log |T|}}$ separated sets $L$ and $R$. Concretely, we obtain disjoint sets $L, R \subseteq T$ with $|L|, |R| = \Omega(|T|)$ so that $\|v_i - v_j\|^2 \geq \Delta$ for any $i \in L$ and $j \in R$.

Next, we define $r$ so that  both 
$L^+ = \{i \in L \mid |v_0 - v_i|^2 \leq r^2\}$ and $L^{-} = \{i \in L \mid |v_0 - v_i|^2 \geq r^2\}$ have more than $\frac{|L|}{2}$ vertices. Note that such an $r$ always exists. Once we fix $r$, we define $R^+ = \{i \in R \mid |v_0 - v_i|^2 \leq r^2\}$ and $R^{-} = \{i \in R \mid |v_0 - v_i|^2 \geq r^2\}$ similarly. Finally, if $|R^+| > \frac{|R|}{2}$ we compute a directed min-cut $F$ between $R^+$ and $L^-$, else we compute a directed min-cut $F$ between $L^+$ and $R-$. 

\begin{lemma}\label{lemma:approx}
$|F| \leq \OO(\frac{1}{\Delta} SDP)$ and $F$ is a $b^*$-balanced separator with respect to $T$ for some $b^* < 1$ depending on $c$, where SDP is the optimal SDP value.
\end{lemma}

\begin{proof}
Without loss of generality, we assume that $|R^+| > \frac{|R|}{2}$. The other case is symmetric. First, observe that for any $i \in R^+$, $j \in L^-$,  we have $d(i,j) > \Delta$. This follows from the definition of the sets $R^+$ and $L^-$ which in particular means that $|v_i - v_0|^2 \leq r^2$ and  $|v_j - v_0|^2 \geq r^2$ and hence $d(i,j) = \|v_i - v_j\|^2 +\|v_j - v_0\|^2 - \|v_i - v_0\|^2 \geq \|v_i - v_j\|^2 \geq \Delta$.

To show the bound on number of cut edges, we analyze the algorithm as follows. For sets $U$ and $W$, we define $d(U,W) := \min_{u \in U, w \in W} d(u,w)$. Similarly, one can define $d(U,w) := d(U,\{w\})$ and $d(u,W) = d(\{u\},W)$ for $u,w \in [n] \cup \{0\}$. Pick a random threshold $\tau \in (0, \frac{\Delta}{4})$ at random and consider the set $X = \{x \mid d(R^+, x) \leq \tau\}$. Clearly $X$ includes $R^+$, and does not include $L^{-}$. Also for each arc $(a,b)$ such that $a \in X$ and $b \notin X$, we have $d(a,b) \geq d(R^+, b) - d(R^+, a)$ since the directed distance metric $d$ satisfies the triangle inequality. It follows that the probability that the edge $(a,b)$ is cut is at most $\OO(\frac{d(a,b)}{\Delta})$ and hence the number of edges cut in expectation is at most $\OO(\frac{SDP}{\Delta})$. In particular, the $(R^+, L^-)$ min-cut has size at most $\OO(\frac{SDP}{\Delta})$.   

Next, we show that this cut is $b^*$- balanced with respect to $T$, for some $b^* < 1$ depending only on $b$. This follows directly from the description of the algorithm. Recall that the ARV separation theorem ensures that $|L|, |R| \geq \Omega_c(|T|)$. Since $|R^+| > \frac{|R|}{2}$ and by construction $|L^-| > \frac{|L|}{2}$ and there is no path from $R^+$ to $L^-$ in $G \setminus F$, the balance condition follows.
\end{proof}

\begin{proof}[Proof of ~\Cref{thm:approxbalsep}]
Given the directed graph $G$ which forms the instance of \DB{}, we first guess the value of $k$, so that there is a set of edges $F$ with $|F| \leq k$ so that each strongly connected component of $G \setminus F$ has size at most $bn$. Using a prefix of the strongly connected components of $G \setminus F$, it follows that there is a partition $X \cup Y$ of $V(G)$ so that $|X|, |Y| \geq c'n$ where $c' = \frac{1-b}{2}$ and there is no arc from $X$ to $Y$ in $G \setminus F$. Next, construct a $(\epsilon, k)$ sample set $T$ for some constant $\epsilon \ll c'$. With constant probability, $T$ is a sample set and $|T \cap X|, |T \cap Y| \geq (c' - \OO(\epsilon))|T| \geq \frac{c'}{2}|T|$.  Use the above algorithm and~\Cref{lemma:approx} with $c = \frac{c'}{2}$ to obtain a set of edges $F'$ with $|F'| \leq \OO(k \sqrt{\log |T|})$ that forms a balanced separator for $T$. It is clear that in $G \setminus F'$, every strongly connected component has at most $b^*|T|$ terminals where $b^*$ depends only on $c'$, as given by~\Cref{lemma:approx}. We choose $\epsilon$ small enough so that $b^* < 1- 2\epsilon$. It follows from~\Cref{def:sample_set} that each strongly connected component has size at most $(b^* + \epsilon)n$. Finally, note that $|T| = \OO(k)$ (since $\epsilon$ is a constant), hence $\sqrt{\log |T|} = \OO(\sqrt{\log k})$. This concludes the proof.

\end{proof}

\section{Vertex Cut Sparsifiers}

In this section, we present our result on vertex cut sparsifiers. We start with a few definitions.

\begin{definition}
Given a directed graph $G$, a set of terminals $T \subseteq V(G)$, and an integer $c$, a directed graph $G'$ satisfying $T \subseteq V(G')$ is said to be a $(c,T)$ vertex cut sparsifier for $G$ if for any partition $A \cup B$ of $T$ such that the size of the minimum $A$-$B$ (vertex) cut in $G$ is at most $c$, the size of the $A$-$B$ mincut is the same in $G$ and $G'$.
\end{definition}

Our next result gives our deterministic construction of vertex sparsifiers in directed graphs.
\mimicking*

To prove this result, we use our result on important separators combined with the standard approach of ``closing'' unnecessary vertices - given a graph $G$ and a vertex $v \in V(G)$, applying the closure operation to $v$ gives the graph $G'$ with the vertex set $V(G) \setminus \{v\}$. The arc set of $G'$ is the same as that of $G$, except that it excludes all arcs incident on $v$, and for every pair of vertices $(u,w)$ such that $u$ is an in-neighbour of $v$ and $w$ is an out-neighbour of $v$,
we add the arc $(u,w)$ to $E(G')$. Before we prove this result, we have the following basic observation. Recall that given a graph $H$ and sets $S,Z \subseteq V(H)$, $R_H^S(Z)$ denotes the set of vertices reachable from some vertex of $S$ in $H \setminus Z$.

\begin{lemma}\label{lemma:closurereach}
Let $A,Y \subseteq V(G)$ and let $G'$ be the graph obtained by applying the closure operation on a vertex $v \in V(G)$. Also suppose $v \notin Y$. Then,

\begin{enumerate}
\item $R_{G'}^A(Y) \subseteq R^A_{G}(Y)$.
\item $R_G^A(Y) \setminus \{v\} \subseteq R^A_{G'}(Y)$

\end{enumerate}

\end{lemma}

\begin{proof}

To prove the first part, consider a $y \in R^A_{G'}(Y)$. Consider the path from some vertex $a \in A$ to $y$ in $G' \setminus Y$. If this path uses an arc $(z_1,z_2)$ of $G'$ that is not an arc of $G$, then we use $v$ to replace this arc by the $3$-path $z_1vz_2$ and obtain a new path from $a$ to $y$ in $G \setminus Y$.

For the second part, similarly consider a $y \in R^A_{G}(Y) \setminus \{v\}$. Consider the path from some $a \in A$ to $y$ in $G \setminus Y$. If this path uses some arc that goes through $v$, consider the predecessor $z_1$ and successor $z_2$ of $v$ in this path: there must be an arc between $z_1$ and $z_2$ by the closure operation. We replace the 3-path $z_1vz_2$ by the arc $(z_1,z_2)$ to obtain a path from $a$ to $y$ in $G' \setminus Y$.

\end{proof}

The next two lemmas show that applying the closure operation to a single vertex which is not in any $A$-$B$ important separator of size at most $c$, across all partitions $A \cup B$ of $T$, preserves the set of important separators. This in turn will imply that we can apply the closure operation at once to all such vertices.  

\begin{lemma}\label{lemma:mimicking1}
Given a directed graph $G$, disjoint sets $A$, $B$ and a vertex $v \notin A \cup B$ which is not in any $A$-$B$ important separator of size $\leq c$, let $G'$ be the graph obtained by closing $v$. Then if $X \subseteq V(G')$, $|X| \leq c$, is an important $A$-$B$ separator in $G'$, then it is an $A$-$B$ important separator in $G$ as well.
\end{lemma}
\begin{proof}
We need to prove three things. First, that $X$ is an $A$-$B$ separator in $G$. Second, it is inclusion-wise minimal. Third, that $X$ is not \emph{dominated} by any other separator - that is, there does not exist another separator $Y$ with $|Y| \leq |X|$ so that $R_{G}^A(X) \subset R_{G}^A(Y)$.

First, $X$ is a separator in $G$, for if not, an $a$-$b$ path for dome $a \in A$ and $b \in B$ in $G \setminus X$ would imply such a path in $G' \setminus X$ by~\Cref{lemma:closurereach}, a contradiction.

Next, if $X \setminus \{w\}$ for some $w \in X$ is also an $A$-$B$ separator in $G$, then it must be an $A$-$B$ separator in $G'$ as well: for if there is a $a$-$b$ path in $G' \setminus (X \setminus\{w\})$ for some $w \in X,a  \in A, b \in B$, then there must be an $a$-$b$ path in $G \setminus (X \setminus \{w\})$ as well by~\Cref{lemma:closurereach}.

Finally, suppose $Y$ is an important separator that dominates $X$ in $G$. Observe that $v \notin Y$. First we show that $Y$ is a separator in $G'$. This follows since an $a$-$b$ path in $G' \setminus Y$ implies an $a$-$b$ path in $G$ by~\Cref{lemma:closurereach}.

Now we wish to prove that $Y$ dominates $X$ in $G'$ as well. Consider a vertex $u \in R^{G'}_A(X)$. Then by~\Cref{lemma:closurereach} it is clear that $u \in R^G_A(X)$ as well. But this means $u \in R^G_A(Y)$, by our assumption. This in turn means that $u \in R^{G'}_A(Y)$. Thus $R^{G'}_A(X) \subseteq R^{G'}_A(Y)$. Finally, we note that $X$ was an inclusion-wise minimal separator in $G'$, and $Y \neq X$, therefore there must exist at least one vertex of $X$ that is reachable in $R^{G'}_A(Y)$. This means that in fact $R^{G'}_A(X) \subset R^{G'}_A(Y)$, and we obtain that $Y$ dominates $X$ in $G'$, a contradiction to the fact that $X$ was important in $G'$.
\end{proof}

The next lemma is analogous and proves the other direction.

\begin{lemma}\label{lemma:mimicking2}
Given a directed graph $G$, disjoint sets $A$, $B$ and a vertex $v$ which is not in any $A-B$ important separator of size $\leq c$, let $G'$ be the graph obtained by closing $v$. Then if $X \subseteq V(G)$, $|X| \leq c$ is an $A$-$B$ important separator in $G$, then it is an $A$-$B$ important separator in $G'$ as well.
\end{lemma}

\begin{proof}

We need to prove three things. First, that $X$ is an $A$-$B$ separator in $G'$. Second, it is inclusion-wise minimal. Third, that $X$ is not dominated by any other separator - that is, there does not exist another separator $Y$ with $|Y| \leq |X|$ so that $R_{G'}^A(X) \subset R_{G'}^A(Y)$.

First, we show that $X$ is an $A$-$B$ separator in $G'$. If not, there is an $a$-$b$ path in $G' \setminus X$ for some $a \in A, b \in B$. But then there is an $a$-$b$ path in $G \setminus X$ by~\Cref{lemma:closurereach}, a contradiction.

Second, suppose $X \setminus \{w\}$ is an $A$-$B$ separator in $G'$ for some $w \in X$. Then $X \setminus \{w\}$ is also an $A$-$B$ separator in $G$, for if there is an $a$-$b$ path in $G \setminus (X \setminus \{w\})$, then such a path exists in $G' \setminus (X \setminus \{w\})$ as well by~\Cref{lemma:closurereach}, again a contradiction.

Third, suppose that $X$ is dominated by $Y$ in $G'$. First, we show that $Y$ is a separator in $G$. This holds since any $a$-$b$ path in $G \setminus Y$ would imply an $a$-$b$ path in $G' \setminus Y$ by~\Cref{lemma:closurereach}.

Then we would like to show that $X$ is dominated by $Y$ in $G$ as well. Pick some $u \in R_G(X)$. If $u \neq v$, then it is clear that $u \in R_{G'}(X)$ as well by~\Cref{lemma:closurereach}. Then $u \in R_{G'}(Y)$ by our assumption, and it follows that $u \in R_G(Y)$ by~\Cref{lemma:closurereach}. Now suppose $v \in R_G(X)$. Pick any in-neighbour $v^*$ of $v$ in $G$, so that $v^*$ in $R_G(X)$. Then it follows that $v^* \in R_{G}(Y)$ by our previous argument. But $v$ is an out-neighbour of $v^*$, and $v \notin Y$. Hence it follows that $v \in R_G(Y)$. Thus it follows that $R_G(X) \subseteq R_G(Y)$. Also since $X$ was an inclusion-wise minimal separator in $G$ and $Y \neq X$, at least one vertex of $X$ must be reachable in $G \setminus Y$, which means that in fact $R_G(X) \subset R_G(Y)$ which contradicts the fact that $X$ is an important separator in $G$.
\end{proof}

\begin{proof}[Proof of~\Cref{thm:mimicking}]
Given $T$, consider all pairs of subsets $(A',B')$ of $T$ satisfying $|A'| \leq k$ and $|B'| \leq 2k$. Let $Y$ be the union of all $(A',B')$ important separators in $G$ across all such $(A',B')$. Observe that $|Y| \leq 4^k {|T| \choose \leq 3k} {3k \choose \leq k}$. 

Apply the closure operation to every vertex of $ Z = V(G) \setminus Y$ simultaneously. More precisely, define $Z_{in} = \{v \in V(G) \mid (v,z) \in E(G)\,\text{for some}\,z \in Z\}$ and $Z_{out} = \{v \in V(G) \mid (z,v) \in E(G)\,\text{for some}\,z \in Z\}$. Then in $G$, we delete the set of vertices $Z$, and add arcs $(z_{in}, z_{out})$, $z_{in} \in Z_{in}, z_{out} \in Z_{out}$ to obtain $G'$.

The correctness and running time of the algorithm follows from the soundness of~\Cref{thm:importantbound},\\~\Cref{lemma:mimicking1} and~\Cref{lemma:mimicking2}.
\end{proof}
}

%% file: appendix.tex
\appendix

\section{Appendix}

\lowerbound*
\begin{proof} 
Essentially, we will construct undirected graphs: these graphs can be made directed simply by adding for every undirected edge $(u,v)$, the directed edges $(u,v)$ and $(v,u)$.

For the first part, consider the star on $c + 1$ vertices for some $c \gg k$. Let $S$ consist of the single center vertex, and let $T$ consist of all the leaves. Observe that for any $T' \subseteq T$ of size at most $k$, the separator which consists of exactly the vertices $T'$ is an important $S$-$T'$ separator of size at most $k$. Thus there are at least ${|T| \choose \leq k}$ such separators.

For the second part, consider again a star, but this time we blow up the center vertex to a clique of size $k + 1$. Formally, consider the graph $G$ on $c + k + 1$ vertices, where $T \subseteq V(G)$ is a clique of size $k + 1$, and $S \subseteq V(G)$ forms the remaining $c$ vertices. Each vertex of $S$ is adjacent to each vertex of $T$, and not adjacent to any vertex of $S$. Thus the graph is a star with a core $T$ of size $k + 1$ and $c$ leaves which form the set $S$. Now for every subset $S' \subseteq S$ with $|S'| \leq k$, $S'$ itself forms an $S'$-$T$ important separator of size at most $k$. Thus there are at least ${|S| \choose \leq k}$ such separators.    
\end{proof}

\begin{lemma}[Tight example for important separator preservation]\label{lemma:impseppreserve}
There is a graph $G$, source vertex $s \in V(G)$, integer $k$, sink vertices $B \subseteq V(G)$ with $|B| = 2k$ and an $s$-$B$ important separator $X \subseteq V(G)$ with $|X| = k$, such that $X$ is not an $s$-$B'$ important separator for any $B' \subset B$.
\end{lemma}

\begin{proof}
Consider the graph in~\Cref{fig:three-boxes}. $X$ is an important $s$-$B$ separator of size $3$. However, clearly $X$ is not an important $s$-$B'$ important separator for any $B' \subset B$, since we must have $B' \cap \{b_{2i - 1}, b_{2i}\} \leq 1$ for some $i \in [3]$ and this means that $(X \setminus \{u_i\}) \cup v_{2i - 1}$ or $(X \setminus \{u_i\}) \cup v_{2i}$ is an $s$-$B'$ separator closer to $B'$ than $X$, contradicting the importance of $X$.
\begin{figure}[!htbp]
  \centering
  \begin{tikzpicture}[node distance=1cm,
      vertex/.style={circle, draw, minimum size=6mm},
      arrow/.style={-{Latex[length=3mm]}},
      group/.style={draw, rectangle, dashed, inner sep=4mm}]
      
    \node[vertex] (s) at (0,0) {\(s\)};
    
    \node[vertex] (u1) at (2,2)   {\(u_1\)};
    \node[vertex] (u2) at (2,0)   {\(u_2\)};
    \node[vertex] (u3) at (2,-2)  {\(u_3\)};
    \draw[arrow] (s) -- (u1);
    \draw[arrow] (s) -- (u2);
    \draw[arrow] (s) -- (u3);

    \node[vertex] (v1) at ($(u1)+(2,0.5)$)  {\(v_1\)};
    \node[vertex] (v2) at ($(u1)+(2,-0.5)$) {\(v_2\)};
    \node[vertex] (v3) at ($(u2)+(2,0.5)$)  {\(v_3\)};
    \node[vertex] (v4) at ($(u2)+(2,-0.5)$) {\(v_4\)};
    \node[vertex] (v5) at ($(u3)+(2,0.5)$)  {\(v_5\)};
    \node[vertex] (v6) at ($(u3)+(2,-0.5)$) {\(v_6\)};
    \draw[arrow] (u1) -- (v1); \draw[arrow] (u1) -- (v2);
    \draw[arrow] (u2) -- (v3); \draw[arrow] (u2) -- (v4);
    \draw[arrow] (u3) -- (v5); \draw[arrow] (u3) -- (v6);

    \node[vertex] (b1) at ($(v1)+(2,0)$) {\(b_1\)};
    \node[vertex] (b2) at ($(v2)+(2,0)$) {\(b_2\)};
    \node[vertex] (b3) at ($(v3)+(2,0)$) {\(b_3\)};
    \node[vertex] (b4) at ($(v4)+(2,0)$) {\(b_4\)};
    \node[vertex] (b5) at ($(v5)+(2,0)$) {\(b_5\)};
    \node[vertex] (b6) at ($(v6)+(2,0)$) {\(b_6\)};
    \draw[arrow] (v1) -- (b1); \draw[arrow] (v2) -- (b2);
    \draw[arrow] (v3) -- (b3); \draw[arrow] (v4) -- (b4);
    \draw[arrow] (v5) -- (b5); \draw[arrow] (v6) -- (b6);

    \begin{pgfonlayer}{background}
      \node[group, fit=(s)] (gS) {};
      \node[above] at (gS.north) {\(S\)};
      \node[group, fit=(u1) (u2) (u3)] (gU) {};
      \node[above] at (gU.north) {\(X\)};
      \node[group, fit=(b1) (b2) (b3) (b4) (b5) (b6)] (gB) {};
      \node[above] at (gB.north) {\(B\)};
    \end{pgfonlayer}

  \end{tikzpicture}
  \caption{Tight example for important separator preservation.}
  \label{fig:three-boxes}
\end{figure}
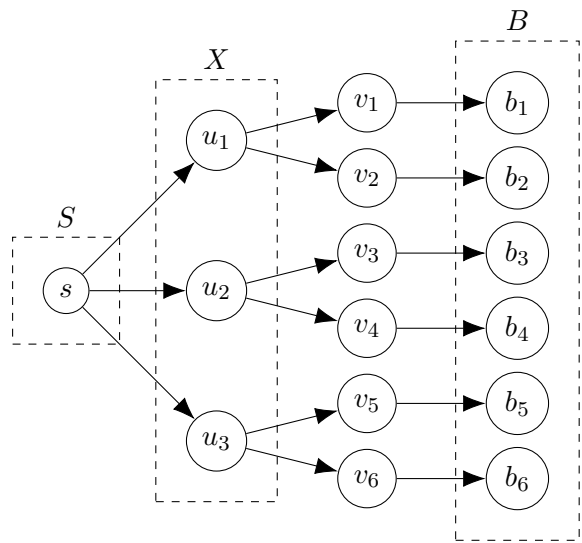

\end{proof}